\documentclass[twoside,11pt]{article}
\pdfoutput=1
\usepackage{amsmath,amsthm,amssymb,amsfonts,eucal,bussproofs}

\usepackage[colorlinks=true,linkcolor=blue,citecolor=blue,urlcolor=blue]{hyperref}
\usepackage[all]{hypcap} 
\usepackage[paperwidth=158mm, paperheight=228mm, 
            centering, textwidth=118mm, lines=38, 
            headsep=5.5mm, top=25mm]{geometry}

\usepackage[center,tiny]{titlesec}
\titlelabel{\thetitle.\ }

\usepackage{fancyhdr}
\newcommand{\LukA}{\textnormal{\L}\forall}
\newcommand{\RPLA}{\textnormal{RPL}\forall}
\newcommand{\GLukA}{\textnormal{G\L}\forall}
\newcommand{\GaLukA}{\textnormal{G}_{\textnormal{a}}\textnormal{\L}\forall}
\newcommand{\GzeroLukA}{\textnormal{G}^0\textnormal{\L}\forall}
\newcommand{\GoneLukA}{\textnormal{G}^1\textnormal{\L}\forall}
\newcommand{\GonezLukA}{\textnormal{G}^{\hat{1}}\textnormal{\L}\forall}
\newcommand{\GtwoLukA}{\textnormal{G}^2\textnormal{\L}\forall}
\newcommand{\GthreeLukA}{\textnormal{G}^3\textnormal{\L}\forall}
\newcommand{\GiLukA}{\textnormal{G}^i\textnormal{\L}\forall}
\newcommand{\HLukA}{\textnormal{H\L}\forall}
\newcommand{\HRPLA}{\textnormal{HRP}\forall}
\newcommand{\HcLukA}{\widehat{\textnormal{H}}\textnormal{\L}\forall}
\newcommand{\HcRPLA}{\widehat{\textnormal{H}}\textnormal{RP}\forall}
\newcommand{\RePl}[3]{[#1]^{#2}_{#3}} 
\newcommand{\SpV}[1]{\mathfrak{#1}} 
\newcommand{\MsV}[1]{\Vert{#1}\Vert} 

\newtheoremstyle{LTheoremStyle}{2mm}{2mm}{\slshape}{}{\mdseries\scshape}{.}{ }{}
\theoremstyle{LTheoremStyle}
\newtheorem{ltheorem}{Theorem}
\newtheorem{llemma}[ltheorem]{Lemma}
\newtheorem{lproposition}[ltheorem]{Proposition}
\newtheorem{lcorollary}[ltheorem]{Corollary}

\theoremstyle{remark}
\newtheorem{remark}{Remark}

\begin{document}
\thispagestyle{empty}

\begin{center} \bf \Large 
COMPARING SEVERAL CALCULI FOR FIRST-ORDER INFINITE-VALUED {\L}UKASIEWICZ LOGIC

\bigskip

Alexander S. Gerasimov
\end{center}
\begin{center}
\href{mailto:alexander.s.gerasimov@ya.ru}
     {\textcolor{black}{\texttt{alexander.s.gerasimov@ya.ru}}}
\end{center}

\bigskip\medskip

\begin{minipage}{0.9\textwidth} \footnotesize
\textbf{Abstract.} 
From the viewpoint of provability, we compare 
some Gentzen-type hypersequent calculi 
for first-order infinite-valued {\L}ukasiewicz logic and 
for first-order rational Pavelka logic with each other and 
with H{\'a}jek's Hilbert-type calculi for these logics.
The key aspect of our comparison is a density elimination proof
for one of the hypersequent calculi considered.

\end{minipage}

\medskip

\section{Introduction}

Mathematical fuzzy logics provide formal foundations for approximate reasoning.
Among the most important such logics are 
first-order in\-fi\-nite-valued {\L}ukasiewicz logic $\LukA$ and 
its expansion by rational truth constants,
first-order rational Pavelka logic $\RPLA$
(see \cite{Hajek1998, HMFL2011,HMFL2015}).

For the logic $\LukA$, as well as for the logic $\RPLA$, 
besides equivalent Hilbert-type calculi (see, e.g., \cite{Hajek1998}), 
only the Gentzen-type calculi mentioned below are known.

For $\LukA$, the article \cite{BaazMetcalfe2010} presents
an analytic hypersequent calculus $\GLukA$ with structural inference rules,
and establishes that $\GLukA$ extended with the cut rule 
and a Hilbert-type calculus for $\LukA$ from \cite{Hajek1998}
prove the same $\LukA$-sentences.

With the aim of developing proof search methods for $\LukA$ and $\RPLA$,
we introduced the following calculi.
First, excluding all the structural inference rules from $\GLukA$, 
in \cite{Ger2017} we obtained a cumulative\footnote{
   We say that
   a hypersequent calculus is \emph{cumulative} if so is its every rule; 
   and a hypersequent rule is \emph{cumulative} if, for its every application, 
   each premise contains the conclusion
   (cf. \cite[Section 3.5.11]{Troelstra2000}).
} 
hypersequent calculus $\GoneLukA$ for $\RPLA$, and 
showed that any $\GLukA$-provable sentence is $\GoneLukA$-provable.
(Also, in \cite{Ger2017} we introduced a variant $\GtwoLukA$ of $\GoneLukA$, 
which is suitable for bottom-up proof search for prenex $\RPLA$-sentences.)
Next, in \cite{Ger201t} we presented a hypersequent calculus $\GthreeLukA$ 
for $\RPLA$ without structural inference rules; 
this calculus is repetition-free, in the sense that 
designations of multisets of formulas are not repeated 
in any premise of its rules.
As shown in \cite{Ger201t}, $\GthreeLukA$ is well-sui\-ted to 
bottom-up proof search for arbitraty $\RPLA$-sentences, and
any $\GoneLukA$-provable sentence (and so any $\GLukA$-provable sentence)
is $\GthreeLukA$-provable.

In the present article, from the viewpoint of provability, 
we compare $\GthreeLukA$ with $\GLukA$ in more detail,
and compare $\GthreeLukA$ with Hilbert-type calculi for $\RPLA$ and $\LukA$ 
from \cite{Hajek1998}.
The key part of our comparison is a proof of the admissibility of
some variants of the density rule for an auxiliary hypersequent calculus;
the features and value of the proof are discussed in the concluding section,
in the context of works related to density elimination.

The article is organized as follows.
In Section \ref{sec:Prelim} we describe the syntax and semantics of the logics
$\LukA$ and $\RPLA$; then we formulate the calculi $\GLukA$ and $\GthreeLukA$, 
as well as the calculus $\GoneLukA$ and a new calculus $\GzeroLukA$,
which help us to compare $\GLukA$ and $\GthreeLukA$.
In Section \ref{sec:InitRelat} we show 
that $\GzeroLukA$ is a conservative extension of $\GLukA$ and
that any $\GzeroLukA$-provable sentence is $\GoneLukA$-provable
(and hence $\GthreeLukA$-provable).
In Section \ref{sec:FurtherRelat}
we establish the admissibility in $\GzeroLukA$ of two variants 
of the density rule (they underlie some rules of $\GthreeLukA$),
and using this, show that $\GthreeLukA$ and $\GzeroLukA$ are equivalent;
hence we conclude that $\GthreeLukA$ is a conservative extension of $\GLukA$.
In Section \ref{sec:HilRelat} we formulate Hilbert-type calculi 
$\HRPLA$ and $\HLukA$ for $\RPLA$ and $\LukA$, respectively; 
next we establish that any $\HRPLA$-provable sentence is provable in $\GthreeLukA$
extended with the cut rule (on $\RPLA$-formulas), 
and that any $\LukA$-sentence is provable in $\HLukA$ iff it is 
provable in $\GthreeLukA$ extended with the cut rule on $\LukA$-formulas.
In Section \ref{sec:Concl} we summarize our results, discuss our 
proof of density admissibility, and pose some problems for further research.

\section{Preliminaries}
\label{sec:Prelim}

Let us define $\LukA$- and $\RPLA$-formulas of a given signature
(it may contain predicate and function symbols of any nonnegative arities).
The notion of a \emph{term} is standard.
\emph{Atomic $\LukA$-formulas} are the truth constant $\bar{0}$ and
predicate symbols with terms as their arguments.
\emph{Atomic $\RPLA$-formulas} are atomic $\LukA$-formulas and
truth constants $\bar{r}$ for all positive rational numbers ${r \leqslant 1}$.\,
$\LukA$- and $\RPLA$-\emph{formulas} are built up as usual from atomic
$\LukA$- and $\RPLA$-formulas, respectively, 
using the following \emph{logical symbols}: 
the binary connective $\to$ and the quantifiers $\forall, \exists$.

An \emph{interpretation} ${\langle \mathcal{D}, \mu \rangle}$
of a given signature is defined as in classical logic, except that 
the map $\mu$ takes each $n$-ary predicate symbol $P$ to
a predicate ${\mu(P): \mathcal{D}^n \to [0,1]}$,\,
where $[0,1]$ is an interval of real numbers.
Let ${M = \langle \mathcal{D}, \mu \rangle}$ be an interpretation. 
Then an $M$-\emph{valuation} is a map of the set of all individual variables 
to~$\mathcal{D}$.
For an $M$-valuation $\nu$, an individual variable $x$, and 
an element ${d\in\mathcal{D}}$, by $\nu[x\mapsto d]$ 
we denote the $M$-valuation that may differ from $\nu$ only on $x$ and 
obeys the condition \,${\nu[x\mapsto d](x) = d}$.

The \emph{value $|t|_{M,\nu}$ of a term $t$} under an interpretation $M$ and 
an $M$-valuation $\nu$ is defined in the standard manner.
The \emph{truth value $|C|_{M,\nu}$ of an $\RPLA$-formula $C$} 
under an interpretation ${M = \langle \mathcal{D}, \mu \rangle}$ and 
an $M$-valuation $\nu$ is defined as follows:

(1)~${|\bar{r}|_{M,\nu}=r}$;

(2)~${|P(t_1,\ldots,t_n)|_{M,\nu} = \mu(P)(|t_1|_{M,\nu},\ldots,|t_n|_{M,\nu})}$ 
for an $n$-ary predicate symbol $P$ and terms $t_1,\ldots,t_n$;

(3)~${|A \to B|_{M,\nu} = \min(1-|A|_{M,\nu}+|B|_{M,\nu}, \,1)}$;

(4)~${|\forall x A|_{M,\nu} = \inf_{d\in\mathcal{D}} |A|_{M,\nu[x\mapsto d]}}$;

(5)~${|\exists x A|_{M,\nu} = \sup_{d\in\mathcal{D}} |A|_{M,\nu[x\mapsto d]}}$.

An $\RPLA$-formula $C$ is called \emph{valid} 
(also written ${\vDash C}$) if ${|C|_{M,\nu} = 1}$ 
for every interpretation $M$ and every $M$-valuation~$\nu$.

The result of substituting a term $t$ for all free occurrences of an individual 
variable $x$ in an $\RPLA$-formula $A$ is denoted by ${\RePl{A}{x}{t}}$.
The provability (resp. unprovability) of an object $\alpha$ 
in a calculus $\mathfrak{C}$\, is written as \,${\vdash_\mathfrak{C} \alpha}$ 
(resp. ${\nvdash_\mathfrak{C} \alpha}$).
By a proof in a calculus, we mean a proof tree. 
In depicting a proof tree $D$, if we place a designation over a node $N$ of $D$
and do not separate the designation from $N$ by a horizontal line, 
then we regard the designation as one for the proof tree whose root is~$N$
and that is a subtree of $D$.

The letters $k$, $l$, $m$, $n$ stand for nonnegative integers.
An expression $k..n$ denotes the set
${\{ k, k+1, \ldots, n\}}$ if ${k \leqslant n}$, and the empty set otherwise.

In what follows, we work with a fixed signature that includes 
a countably infinite set of nullary function symbols called \emph{parameters}.

Let us formulate the auxiliary hypersequent calculus $\GzeroLukA$ and
define accompanying notions and notation commom to several calculi considered.

We introduce two countably infinite disjoint sets of new words and
call such words \emph{semipropositional variables of type}~0 and 
\emph{of type}~1, respectively.
An $\RPLA$-formula as well as a semipropositional variable (of any type) 
is called a \emph{formula}.

An $\RPLA^1_0$-\emph{sequent} (or simply a \emph{sequent}) 
is written as ${\Gamma\Rightarrow\Delta}$ and
is an ordered pair of finite multisets $\Gamma$ and $\Delta$ 
consisting of formulas.
An $\RPLA^1_0$-\emph{hypersequent} (\emph{hypersequent} for short)
is a finite multiset of sequents and is written as 
${\Gamma_1\Rightarrow\Delta_1\,|\ldots|\,\Gamma_n\Rightarrow\Delta_n}$ or
${\big[ \Gamma_i\Rightarrow\Delta_i \big]_{i\in 1..n}}$.

A sequent not containing logical symbols is called \emph{atomic}.
Suppose that $\mathcal{H}$ is a hypersequent; then by $\mathcal{H}_{at}$ 
we denote the hypersequent obtained from $\mathcal{H}$ by removing 
all non-atomic sequents.

We define an \emph{hs-interpretation} as an interpretation
$\langle \mathcal{D}, \mu \rangle$ in which the map $\mu$ additionally takes 
each semipropositional variable of type~0 to a real number in $[0, +\infty)$ and
each semipropositional variable of type~1 to a real number in $(-\infty, 1]$.
For a semipropositional variable $\SpV{p}$,
an hs-interpretation ${M = \langle \mathcal{D}, \mu \rangle}$, and
an $M$-valuation $\nu$,
the value $\mu(\SpV{p})$ will also be written as $|\SpV{p}|_M$ 
and as $|\SpV{p}|_{M,\nu}$.

For a finite multiset $\Gamma$ of formulas, a sequent
${\Gamma\Rightarrow\Delta}$, an hs-in\-ter\-pre\-ta\-tion $M$, 
and an $M$-valuation $\nu$, we put
$${\MsV{\Gamma}_{M,\nu} = \sum_{A\in\Gamma} (|A|_{M,\nu}-1)},$$
where the summation is performed taking multiplicities of multiset elements
into account, and \,${\sum_{A\in\varnothing} (\ldots) = 0}$.
A sequent ${\Gamma\Rightarrow\Delta}$ is called \emph{true} 
under an hs-interpretation $M$ and an $M$-valuation $\nu$\, if
$$\MsV{\Gamma}_{M,\nu} \leqslant \MsV{\Delta}_{M,\nu}.$$ 

Following \cite[Definition 1]{BaazMetcalfe2010}, we say that
a hypersequent $\mathcal{H}$ is \emph{valid} (and write ${\vDash \mathcal{H}}$)
if, for every hs-interpretation $M$ and every $M$-va\-lu\-a\-tion $\nu$,
some sequent in $\mathcal{H}$ is true under $M$ and $\nu$. 
Note that, for an $\RPLA$-formula $A$, 
\,${\vDash A}$ iff \,${\vDash (\Rightarrow A)}$.
To denote that a hypersequent $\mathcal{G}$ is not valid, we write
${\nvDash \mathcal{G}}$.

Unless otherwise specified, below
the letters $A$, $B$, and $C$ denote any $\RPLA$-formulas, 
$\Gamma$, $\Delta$, $\Pi$, and $\Sigma$ any finite multisets of formulas, 
$S$ any sequent,
$\mathcal{G}$ and $\mathcal{H}$ any hypersequents, 
$x$ any individual variable, 
$t$ any closed term,
$a$ any parameter,
and $r$ any rational number such that ${0 \leqslant r \leqslant 1}$;
all these letters may have subscripts and superscripts.
Also $\SpV{p}_i$ (${i=0,1}$) denotes any semipropositional variable of type~$i$.

The language of the calculus $\GzeroLukA$ consists of all possible hypersequents.
A hypersequent $\mathcal{H}$ is called an \emph{axiom of $\GzeroLukA$}
if ${\vDash \mathcal{H}_{at}}$.
(Axioms of $\GzeroLukA$ can be recognized by a polynomial algorithm
in much the same way as described in \cite[Section 4.2]{Ger2017}.)

The inference rules of the calculus $\GzeroLukA$ are: 
\begin{center}
$\dfrac{\mathcal{G}\,|\,\Gamma, A\to B \Rightarrow \Delta\,|\,\Gamma \Rightarrow \Delta\,|\,\Gamma,B \Rightarrow A,\Delta}
 {\mathcal{G}\,|\,\Gamma, A\to B \Rightarrow \Delta}~(\to\:\Rightarrow)^0$,\\[\medskipamount]
$\dfrac{\mathcal{G}\,|\,\Gamma \Rightarrow A\to B, \Delta\,|\,\Gamma \Rightarrow \Delta;  \quad  \mathcal{G}\,|\,\Gamma \Rightarrow A\to B, \Delta\,|\,\Gamma,A \Rightarrow B,\Delta}
 {\mathcal{G}\,|\,\Gamma \Rightarrow A\to B, \Delta}~(\Rightarrow\:\to)^0$,\\[\medskipamount]
$\dfrac{\mathcal{G}\,|\,\Gamma, \forall x A \Rightarrow\Delta\,|\,\Gamma, \RePl{A}{x}{t} \Rightarrow\Delta}
 {\mathcal{G}\,|\,\Gamma, \forall x A \Rightarrow\Delta}~(\forall\Rightarrow)^0$,\\[\medskipamount]
$\dfrac{\mathcal{G}\,|\,\Gamma \Rightarrow \forall x A, \Delta\,|\,\Gamma \Rightarrow \RePl{A}{x}{a}, \Delta}
 {\mathcal{G}\,|\,\Gamma \Rightarrow \forall x A, \Delta}~(\Rightarrow\forall)^0$,\\[\medskipamount]
$\dfrac{\mathcal{G}\,|\,\Gamma \Rightarrow \exists x A, \Delta\,|\,\Gamma \Rightarrow \RePl{A}{x}{t}, \Delta}
 {\mathcal{G}\,|\,\Gamma \Rightarrow \exists x A, \Delta}~(\Rightarrow\exists)^0$,\\[\medskipamount]
$\dfrac{\mathcal{G}\,|\,\Gamma, \exists x A \Rightarrow\Delta\,|\,\Gamma, \RePl{A}{x}{a} \Rightarrow\Delta}
 {\mathcal{G}\,|\,\Gamma, \exists x A \Rightarrow\Delta}~(\exists\Rightarrow)^0$,
\end{center}
where $a$ does not occur in the conclusion of ${(\Rightarrow\forall)^0}$ or
${(\exists\Rightarrow)^0}$.

For convenience in comparing calculi, 
we also introduce the calculus $\GonezLukA$ that is obtained from $\GzeroLukA$
by replacing the inference rule ${(\to\:\Rightarrow)^0}$ with
$$\dfrac{\mathcal{G}\,|\,\Gamma, A\to B \Rightarrow \Delta\,|\,\Gamma, \SpV{p}_1 \Rightarrow \Delta\,|\,B \Rightarrow \SpV{p}_1, A}
  {\mathcal{G}\,|\,\Gamma, A\to B \Rightarrow \Delta}~(\to\:\Rightarrow)^{\hat{1}},$$
where $\SpV{p}_1$ does not occur in the conclusion.

The calculus $\GoneLukA$ \cite{Ger2017} is obtained from $\GonezLukA$ 
by restricting the language of $\GonezLukA$ to
hypersequents not containing semipropositional variables of type~0;
such hypersequents are called $\RPLA^1$-\emph{hypersequents}.

The rule of $\GiLukA$ (${i=1,\hat{1}}$) that 
corresponds to a rule of $\GzeroLukA$ is denoted just as the latter 
but with the superscript $i$ instead of~0.

\begin{remark} \label{rem:GonezGone}
It is clear that, for an $\RPLA^1$-hypersequent $\mathcal{H}$,
a $\GonezLukA$-proof of $\mathcal{H}$ is a $\GoneLukA$-proof of $\mathcal{H}$, 
and conversely. 
\end{remark}
 
The calculus $\GthreeLukA$ \cite{Ger201t} is obtained from $\GzeroLukA$ 
by replacing all the inference rules with the following ones:
\begin{center} 
$\dfrac{\mathcal{G}\,|\, \Gamma, \SpV{p}_1 \Rightarrow \Delta\,|\,B \Rightarrow \SpV{p}_1, A}
 {\mathcal{G}\,|\,\Gamma, A\to B \Rightarrow \Delta}~(\to\:\Rightarrow)^3$,\\[\medskipamount]
$\dfrac{\mathcal{G}\,|\, \Gamma \Rightarrow \Delta;  \quad  \mathcal{G}\,|\, \Gamma,A \Rightarrow B,\Delta}
 {\mathcal{G}\,|\,\Gamma \Rightarrow A\to B, \Delta}~(\Rightarrow\:\to)^3$,\\[\medskipamount]
$\dfrac{\mathcal{G}\,|\, \Gamma, \SpV{p}_1 \Rightarrow\Delta\,|\,\forall x A \Rightarrow \SpV{p}_1\,|\,\RePl{A}{x}{t} \Rightarrow \SpV{p}_1}
 {\mathcal{G}\,|\, \Gamma, \forall x A \Rightarrow\Delta}~(\forall\Rightarrow)^3$, \hfill
$\dfrac{\mathcal{G}\,|\, \Gamma \Rightarrow \RePl{A}{x}{a}, \Delta}
 {\mathcal{G}\,|\,\Gamma \Rightarrow \forall x A, \Delta}~(\Rightarrow\forall)^3$,\\[\medskipamount]
$\dfrac{\mathcal{G}\,|\, \Gamma \Rightarrow \SpV{p}_0, \Delta\,|\,\SpV{p}_0 \Rightarrow \exists x A\,|\,\SpV{p}_0 \Rightarrow \RePl{A}{x}{t}}
 {\mathcal{G}\,|\, \Gamma \Rightarrow \exists x A, \Delta}~(\Rightarrow\exists)^3$, \hfill
$\dfrac{\mathcal{G}\,|\, \Gamma, \RePl{A}{x}{a} \Rightarrow\Delta}
 {\mathcal{G}\,|\,\Gamma, \exists x A \Rightarrow\Delta}~(\exists\Rightarrow)^3$,
\end{center}
where $\SpV{p}_1$ does not occur in the conclusion of ${(\to\:\Rightarrow)^3}$
or ${(\forall\Rightarrow)^3}$,
$\SpV{p}_0$ does not occur in the conclusion of ${(\Rightarrow\exists)^3}$,
and $a$ does not occur in the conclusion of ${(\Rightarrow\forall)^3}$ or
${(\exists\Rightarrow)^3}$.

For an application of an inference rule of $\GiLukA$ (${i=0,1,\hat{1},3}$),
the \emph{principal} formula occurrence and 
the \emph{principal} sequent occurrence 
are defined in essentially the same manner as in 
\cite[\S~49]{Kleene1967} and \cite[Section 3.5.1]{Troelstra2000}.
The notion of an \emph{ancestor} of a sequent occurrence in 
a $\GiLukA$-proof (${i=0,1,\hat{1},3}$) is defined much as 
the notion of an ancestor of a formula occurrence is defined 
in \cite[\S~49]{Kleene1967}.  

Now we formulate the calculus $\GLukA$ \cite{BaazMetcalfe2010},
using parameters instead of free individual variables, which are
syntactically distinct from bound individual variables 
in \cite{BaazMetcalfe2010}.
The language of $\GLukA$ consists of all possible 
$\LukA$-\emph{hypersequents}, i.e., hypersequents that
do not contain semipropositional variables and, of truth constants, 
may contain only~$\bar{0}$.

The axiom schemes of $\GLukA$ are:
${A \Rightarrow A~(\text{id}),  \   \Rightarrow(\Lambda),  \ 
 \bar{0} \Rightarrow A~(\bar{0}\Rightarrow)}$, 
where $A$ is an $\LukA$-formula.

The inference rules of $\GLukA$ are:
\begin{center}
$\dfrac{\mathcal{G}}{\mathcal{G}\,|\, S}$~(ew), \quad
$\dfrac{\mathcal{G}\,|\, S \,|\, S}
       {\mathcal{G}\,|\, S}$~(ec), \quad
$\dfrac{\mathcal{G}\,|\,\Gamma\Rightarrow\Delta}
       {\mathcal{G}\,|\,\Gamma,C\Rightarrow\Delta}$~(wl),\\[\medskipamount]
$\dfrac{\mathcal{G}\,|\,\Gamma_1,\Gamma_2\Rightarrow\Delta_1,\Delta_2}
       {\mathcal{G}\,|\,\Gamma_1\Rightarrow\Delta_1\,|\,\Gamma_2\Rightarrow\Delta_2}$~(split), \quad
$\dfrac{\mathcal{G}\,|\,\Gamma_1\Rightarrow\Delta_1;  \quad  \mathcal{G}\,|\,\Gamma_2\Rightarrow\Delta_2}
       {\mathcal{G}\,|\,\Gamma_1,\Gamma_2\Rightarrow\Delta_1,\Delta_2}$~(mix),\\[\medskipamount]
$\dfrac{\mathcal{G}\,|\,\Gamma,B \Rightarrow A,\Delta}
       {\mathcal{G}\,|\,\Gamma, A\to B \Rightarrow \Delta}~(\to\:\Rightarrow)$, \quad
$\dfrac{\mathcal{G}\,|\,\Gamma \Rightarrow \Delta;  \quad  \mathcal{G}\,|\,\Gamma,A \Rightarrow B,\Delta}
       {\mathcal{G}\,|\,\Gamma \Rightarrow A\to B, \Delta}~(\Rightarrow\:\to)$,\\[\medskipamount]
$\dfrac{\mathcal{G}\,|\,\Gamma, \RePl{A}{x}{t} \Rightarrow\Delta}
       {\mathcal{G}\,|\,\Gamma, \forall x A \Rightarrow\Delta}~(\forall\Rightarrow)$, \quad
$\dfrac{\mathcal{G}\,|\,\Gamma \Rightarrow \RePl{A}{x}{a}, \Delta}
       {\mathcal{G}\,|\,\Gamma \Rightarrow \forall x A, \Delta}~(\Rightarrow\forall)$,\\[\medskipamount]
$\dfrac{\mathcal{G}\,|\,\Gamma \Rightarrow \RePl{A}{x}{t}, \Delta}
       {\mathcal{G}\,|\,\Gamma \Rightarrow \exists x A, \Delta}~(\Rightarrow\exists)$, \quad
$\dfrac{\mathcal{G}\,|\,\Gamma, \RePl{A}{x}{a} \Rightarrow\Delta}
       {\mathcal{G}\,|\,\Gamma, \exists x A \Rightarrow\Delta}~(\exists\Rightarrow)$,
\end{center}
where all the premises and conclusions are $\LukA$-hypersequents, and 
$a$ does not occur in the conslusion of ${(\Rightarrow\forall)}$ or
${(\exists\Rightarrow)}$.
The first five of these rules are called \emph{structural};
the others, \emph{logical}. 

For each calculus formulated above, its every one-premise rule in whose premise
$a$, $t$, or $\SpV{p}_i$ (${i=0,1}$) figurates, and
for any application of the rule, 
the $a$, $t$, or $\SpV{p}_i$ is called, respectively, 
the \emph{proper} parameter, \emph{proper} term, or 
\emph{proper} semipropositional variable of the application.

A \emph{proof of \textnormal{(}for\textnormal{)} an $\RPLA$-formula $A$}
in any of the hypersequent calculi given above is
a proof of the hypersequent \,${\Rightarrow A}$ in the respective calculus.

\section{Initial relationships between\\ the hypersequent calculi considered}
\label{sec:InitRelat}

\begin{ltheorem}  \label{Th:GzeroConsExtGLukA}
$\GzeroLukA$ is a conservative extension of \,$\GLukA$;
i.e., for any $\LukA$-hypersequent $\mathcal{H}$,
\,${\vdash_{\GzeroLukA} \mathcal{H}}$ iff \:${\vdash_{\GLukA} \mathcal{H}}$.
\end{ltheorem}

\begin{proof}[\textsc{Proof}]
Let $\mathcal{H}$ be an $\LukA$-hypersequent.
If ${\vdash_{\GzeroLukA} \mathcal{H}}$, 
then ${\vdash_{\GLukA} \mathcal{H}}$ by Lemma \ref{Lem:GzeroImpGLukA} below.
Let us prove the converse. 

We get the calculus $\GaLukA$ from $\GLukA$ 
by taking $A$ in the axiom schemes (id) and ${(\bar{0}\Rightarrow)}$
to be an atomic $\LukA$-formula.
Lemma 6 in \cite{Ger2016} guarantees that
\,${\vdash_{\GLukA} \mathcal{H}}$ iff \,${\vdash_{\GaLukA} \mathcal{H}}$.

So it suffices to show that \,${\vdash_{\GaLukA} \mathcal{H}}$ implies
\,${\vdash_{\GzeroLukA} \mathcal{H}}$.
Any axiom of $\GaLukA$ is obviously an axiom of $\GzeroLukA$.
All the structural rules of $\GaLukA$ 
are admissible for $\GzeroLukA$ by Lemma \ref{Lem:StructRulesAdmGzero} below.
Since the rule (ew) is admissible for $\GzeroLukA$, it follows easily that 
all the logical rules of $\GaLukA$ are admissible for $\GzeroLukA$.
\end{proof}

\begin{llemma}  \label{Lem:GzeroImpGLukA}
For any $\LukA$-hypersequent $\mathcal{H}$,\,
if \:${\vdash_{\GzeroLukA} \mathcal{H}}$, 
then \,${\vdash_{\GLukA} \mathcal{H}}$.
\end{llemma}

\begin{proof}[\textsc{Proof}]
Let $\mathcal{H}$ be an $\LukA$-hypersequent.

If $\mathcal{H}$ is an axom of $\GzeroLukA$, then
\,${\vDash \mathcal{H}_{at}}$, and
by the completeness of $\GLukA$ for quantifier-free $\LukA$-hypersequents
\cite[Theorem 6.24]{MOG2009}, we get \,${\vdash_{\GLukA} \mathcal{H}_{at}}$,
whence \,${\vdash_{\GLukA} \mathcal{H}}$ by the rule (ew).  

To conclude the proof, it is sufficient to show that 
all the rules of $\GzeroLukA$ are derivable in $\GLukA$ 
if their premises and conclusions are restricted to $\LukA$-hypersequents.
For the rule ${(\to\:\Rightarrow)^0}$, we have:
\begin{center}
\def\ScoreOverhang{0pt}
\AxiomC{$\mathcal{G}\,|\,\Gamma, A\to B \Rightarrow \Delta\,|\,\Gamma \Rightarrow \Delta\,|\,\Gamma,B \Rightarrow A,\Delta$}
\RightLabel{$(\to\:\Rightarrow)$}
\UnaryInfC{$\mathcal{G}\,|\,\Gamma, A\to B \Rightarrow \Delta\,|\,\Gamma \Rightarrow \Delta\,|\,\Gamma, A\to B \Rightarrow \Delta$}
\RightLabel{(wl)}
\UnaryInfC{$\mathcal{G}\,|\,\Gamma, A\to B \Rightarrow \Delta\,|\,\Gamma, A\to B \Rightarrow \Delta\,|\,\Gamma, A\to B \Rightarrow \Delta$}
\RightLabel{$\text{(ec)} \!\times\! 2$.}
\UnaryInfC{$\mathcal{G}\,|\,\Gamma, A\to B \Rightarrow \Delta$}
\DisplayProof
\end{center}
For the rule ${(\forall\Rightarrow)^0}$, we have:
\begin{center}
\def\ScoreOverhang{0pt}
\AxiomC{$\mathcal{G}\,|\,\Gamma, \forall x A \Rightarrow\Delta\,|\,\Gamma, \RePl{A}{x}{t} \Rightarrow\Delta$}
\RightLabel{$(\forall\Rightarrow)$}
\UnaryInfC{$\mathcal{G}\,|\,\Gamma, \forall x A \Rightarrow\Delta\,|\,\Gamma, \forall x A \Rightarrow\Delta$}
\RightLabel{(ec).}
\UnaryInfC{$\mathcal{G}\,|\,\Gamma, \forall x A \Rightarrow\Delta$}
\DisplayProof
\end{center}
The other rules of $\GzeroLukA$ are treated similarly to
${(\forall\Rightarrow)^0}$.
\end{proof}

\begin{llemma}  \label{Lem:StructRulesAdmGzero}
The following rules are admissible for $\GzeroLukA$:
\begin{center}
$\dfrac{\mathcal{G}}{\mathcal{G}\,|\, S}~\textnormal{(ew)}^0$, \quad
$\dfrac{\mathcal{G}\,|\, S \,|\, S}
       {\mathcal{G}\,|\, S}~\textnormal{(ec)}^0$, \quad
$\dfrac{\mathcal{G}\,|\,\Gamma\Rightarrow\Delta}
       {\mathcal{G}\,|\,\Gamma,C\Rightarrow\Delta}~\textnormal{(wl)}^0$,\\[\medskipamount]
$\dfrac{\mathcal{G}\,|\,\Gamma_1,\Gamma_2\Rightarrow\Delta_1,\Delta_2}
       {\mathcal{G}\,|\,\Gamma_1\Rightarrow\Delta_1\,|\,\Gamma_2\Rightarrow\Delta_2}~\textnormal{(split)}^0$, \quad
$\dfrac{\mathcal{G}\,|\,\Gamma_1\Rightarrow\Delta_1;  \quad  \mathcal{G}\,|\,\Gamma_2\Rightarrow\Delta_2}
       {\mathcal{G}\,|\,\Gamma_1,\Gamma_2\Rightarrow\Delta_1,\Delta_2}~\textnormal{(mix)}^0$.
\end{center}
Moreover, the rules $\textnormal{(ew)}^0$, $\textnormal{(ec)}^0$, and
$\textnormal{(split)}^0$ are height-preserving admissible, 
or briefly hp-admissible, for $\GzeroLukA$.
\end{llemma}

\begin{proof}[\textsc{Proof}]
1.~It is clear that $\textnormal{(ew)}^0$ is hp-admissible (for $\GzeroLukA$).

2. To establish the hp-admissibility of $\textnormal{(ec)}^0$, we note
that all the rules of $\GzeroLukA$ are cumulative and 
proceed just as in the proof of Lemma 5 in \cite{Ger2017}.

3. To show that $\textnormal{(wl)}^0$ is admissible, we use induction on
the number of logical symbol occurrences in $C$.
Let ${\mathcal{H}_1 = (\mathcal{G}\,|\,\Gamma\Rightarrow\Delta)}$ and
${\mathcal{H}_2 = (\mathcal{G}\,|\,\Gamma,C\Rightarrow\Delta)}$.
We can assume that 
there is a \mbox{($\GzeroLukA$-)proof} $D_1$ for $\mathcal{H}_1$
such that no proper parameter from $D_1$ occurs in~$C$.

3.1. Suppose that $C$ is atomic or is of the form ${(A\to B)}$.
From $D_1$ we construct a proof search tree $D_2^0$ for $\mathcal{H}_2$ 
as follows.
For each occurrence $\mathcal{S}$ of a sequent of the form
${\Pi\Rightarrow\Sigma}$,\, 
if $\mathcal{S}$ is an ancestor of the distinguished occurrence of 
the sequent ${\Gamma\Rightarrow\Delta}$ in the root of $D_1$, then
we replace $\mathcal{S}$
by an occurrence $\mathcal{S}'$ of the sequent ${\Pi,C\Rightarrow\Sigma}$.
We also mark $\mathcal{S}'$ if 
$\mathcal{S}$ is an atomic sequent occurrence in a leaf of~$D_1$.

If $C$ is atomic, then $D_2^0$ is a proof for~$\mathcal{H}_2$.

Suppose $C$ is of the form ${(A\to B)}$, and
$\mathcal{S}_0,\ldots,\mathcal{S}_{l-1}$ are all distinct
marked sequent occurrences in~$D_2^0$.

We expand $D_2^0$ by performing the following for each ${i=0,\ldots,l-1}$: 
on the only branch $\mathcal{B}_i$ of $D_2^i$ containing $\mathcal{S}_i$,
apply the rule ${(\to\:\Rightarrow)^0}$ backward to 
the ancestor of $\mathcal{S}_i$ in the leaf on $\mathcal{B}_i$,
and denote by $D_2^{i+1}$ the tree obtained as a result of 
this backward application.

Note that if $\mathcal{S}_i$ is an occurrence of a sequent of the form
${\Pi_i,C\Rightarrow\Sigma_i}$, then 
the atomic sequent ${\Pi_i\Rightarrow\Sigma_i}$ is on the continuation of
the branch $\mathcal{B}_i$ in $D_2^{i+1}$.
Therefore, it is easy to see that $D_2^l$ is a proof for~$\mathcal{H}_2$.

3.2. Suppose that $C$ is of the form $Q x A$, where $Q$ is a quantifier.
By the induction hypothesis, 
there is a proof for
${\mathcal{H} = (\mathcal{H}_2 \,|\, \Gamma, \RePl{A}{x}{a} \Rightarrow \Delta)}$,
where $a$ is a parameter not occurring in $\mathcal{H}_2$.
By applying the rule ${(Q\Rightarrow)^0}$ to the distinguished occurrence of
$\RePl{A}{x}{a}$ in $\mathcal{H}$, we get a proof for~$\mathcal{H}_2$.

4. The proof of the hp-admissibility of $\textnormal{(split)}^0$ 
is very similar to the proof of Lemma 7 in \cite{Ger2017}.\footnote{ 
  \label{fn:SplitMix}
  See also Section \ref{secApp:SplitMixAdmGzero} 
  (of the appendix) on p.~\pageref{secApp:SplitMixAdmGzero}.
}

5. The proof of the admissibility of $\textnormal{(mix)}^0$ can be easily 
obtained from the proof of Lemma 8 in \cite{Ger2017} by identifying
the notion of a completable ancestor of a sequent occurrence
with the notion of an ancestor of a sequent occurrence
(the former notion is used in \cite{Ger2017}).$^{\ref{fn:SplitMix}}$
\end{proof}

\begin{ltheorem}  \label{Th:GzeroImpGonez}
If \:${\vdash_{\GzeroLukA} \mathcal{H}}$, 
then \,${\vdash_{\GonezLukA} \mathcal{H}}$.
\end{ltheorem}

\begin{proof}[\textsc{Proof}]
All axioms of $\GzeroLukA$ are axioms of $\GonezLukA$.
All the rules of $\GzeroLukA$, except for the rule ${(\to\:\Rightarrow)^0}$, 
are rules of $\GonezLukA$.
Hence, it suffices to establish that ${(\to\:\Rightarrow)^0}$ is admissible for
$\GonezLukA$.

For this, we use the rules
$$
\dfrac{\mathcal{G}\,|\,\Gamma\Rightarrow\Delta}
  {\mathcal{G}\,|\,\Gamma,\SpV{p}_1 \Rightarrow \SpV{p}_1,\Delta}~(\textnormal{sp}_1\!\!\Rightarrow\!\!\textnormal{sp}_1)^0
\qquad \text{and} \qquad
\dfrac{\mathcal{G}\,|\,\Gamma\Rightarrow\Delta}
  {\mathcal{G}\,|\,\Gamma,\SpV{p}_1\Rightarrow\Delta}~\textnormal{(wl)}^0_{\textnormal{sp}_1},
$$
whose hp-admissibility for $\GonezLukA$ is obvious.
We also use the rules $\text{(ec)}^0$ and $\text{(split)}^0$ 
from Lemma \ref{Lem:StructRulesAdmGzero}, noting that
the proofs of their hp-ad\-mis\-si\-bi\-lity for $\GonezLukA$ 
are entirely analogous to the proofs of Lemmas 5 and 7 in \cite{Ger2017}.

We obtain the conclusion of the rule ${(\to\:\Rightarrow)^0}$ from its premise
by rules, which are admissible for $\GonezLukA$, 
as shown in Figure \ref{fig:GzeroRuleToRightAdmGonez},
where $\SpV{p}_1$ does not occur in 
\,${\mathcal{G}\,|\,\Gamma, A\to B \Rightarrow \Delta}$.
\begin{figure}[!tb]
\begin{center}
\def\ScoreOverhang{0pt}
\AxiomC{$\mathcal{G}\,|\,\Gamma, A\to B \Rightarrow \Delta\,|\,\Gamma \Rightarrow \Delta\,|\,\Gamma,B \Rightarrow A,\Delta$}
\RightLabel{$(\textnormal{sp}_1\!\!\Rightarrow\!\!\textnormal{sp}_1)^0$}
\UnaryInfC{$\mathcal{G}\,|\,\Gamma, A\to B \Rightarrow \Delta\,|\,\Gamma \Rightarrow \Delta\,|\,\Gamma,B,\SpV{p}_1 \Rightarrow \SpV{p}_1,A,\Delta$}
\RightLabel{$\text{(split)}^0$}
\UnaryInfC{$\mathcal{G}\,|\,\Gamma, A\to B \Rightarrow \Delta\,|\,\Gamma \Rightarrow \Delta\,|\,\Gamma,\SpV{p}_1 \Rightarrow \Delta\,|\,B \Rightarrow \SpV{p}_1,A$}
\RightLabel{$\textnormal{(wl)}^0_{\textnormal{sp}_1}$}
\UnaryInfC{$\mathcal{G}\,|\,\Gamma, A\to B \Rightarrow \Delta\,|\,\Gamma,\SpV{p}_1 \Rightarrow \Delta\,|\,\Gamma,\SpV{p}_1 \Rightarrow \Delta\,|\,B \Rightarrow \SpV{p}_1,A$}
\RightLabel{$\text{(ec)}^0$}
\UnaryInfC{$\mathcal{G}\,|\,\Gamma, A\to B \Rightarrow \Delta\,|\,\Gamma,\SpV{p}_1 \Rightarrow \Delta\,|\,B \Rightarrow \SpV{p}_1,A$}
\RightLabel{$(\to\:\Rightarrow)^{\hat{1}}$,}
\UnaryInfC{$\mathcal{G}\,|\,\Gamma, A\to B \Rightarrow \Delta$}
\DisplayProof
\caption{Obtaining the conclusion of the rule ${(\to\:\Rightarrow)^0}$ from its premise.}
\label{fig:GzeroRuleToRightAdmGonez}
\end{center}
\end{figure}
Thus ${(\to\:\Rightarrow)^0}$ is admissible for $\GonezLukA$.
\end{proof}

\begin{ltheorem}  \label{Th:GonezImpGthree}
If \:${\vdash_{\GonezLukA} \mathcal{H}}$, 
then \,${\vdash_{\GthreeLukA} \mathcal{H}}$.
\end{ltheorem}

\begin{proof}[\textsc{Proof}\nopunct]
is obtained from the proofs of Lemma 6 and Theorem 2 in \cite{Ger201t} 
by substituting the superscript $\hat{1}$ for the superscript 1 
(in $\GoneLukA$ and the designations of the rules of $\GoneLukA$).
\end{proof}

\section{The admissibility of variants of the density rule for $\GzeroLukA$\\
         and further relationships between\\ the hypersequent calculi considered}
\label{sec:FurtherRelat}
              
The primary goal of this section is to show that 
if a hypersequent is $\GthreeLukA$-provable, then it is $\GzeroLukA$-provable.
For this, we establish that all the rules of $\GthreeLukA$ are admissible for
$\GzeroLukA$.

As shown in the proof of the next lemma,
the rules ${(\to\:\Rightarrow)^3}$, ${(\forall\Rightarrow)^3}$, and
${(\Rightarrow\exists)^3}$ of $\GthreeLukA$ are based on the rules 
$$
\dfrac{\mathcal{G} \,|\,\Gamma, \SpV{p}_1 \Rightarrow \Delta \,|\,C \Rightarrow \SpV{p}_1}
  {\mathcal{G} \,|\,\Gamma, C \Rightarrow \Delta}~(\text{den}_1)
\quad \text{and} \quad
\dfrac{\mathcal{G} \,|\,\Gamma \Rightarrow \SpV{p}_0, \Delta \,|\,\SpV{p}_0 \Rightarrow C}
  {\mathcal{G} \,|\,\Gamma \Rightarrow C, \Delta}~(\text{den}_0),
$$
where $\SpV{p}_i$ does not occur in the conclusion of $(\text{den}_i)$, ${i=0,1}$.
The last two rules can be characterized as
nonstandard variants of the density rule, cf. \cite[Section~4.5]{MOG2009}.

\begin{remark} \label{rem:DensitySound}
The (standard) \emph{density rule} in the hypersequent formulation is:
$$\dfrac{\mathcal{G} \,|\, \Gamma, \SpV{p} \Rightarrow \Delta \,|\, \Pi \Rightarrow \SpV{p}, \Sigma}
  {\mathcal{G} \,|\, \Gamma, \Pi \Rightarrow \Delta, \Sigma}~(\text{den}),$$
where $\SpV{p}$ is a propositional variable not occurring in the conclusion;
see \cite[Section 4.5]{MOG2009}.
Given our definition of the validity of a hypersequent,
it is not hard to check that (den) is unsound, but becomes sound if
we expand the notion of a hypersequent by special variables 
interpreted by any real numbers, and require $\SpV{p}$ to be such a variable
not occurring in the conclusion.\footnote{
 See also Section \ref{secApp:DenSound} (of the appendix)
 on p.~\pageref{secApp:DenSound}.
}
Let us refer to this modified rule (den) as the \emph{nonstandard density rule}.
\end{remark}

\begin{llemma}  \label{Lem:IfDeniAdmGzero}
If the rules $(\textnormal{den}_1)$ and $(\textnormal{den}_0)$
are admissible for $\GzeroLukA$, then
\,${\vdash_{\GthreeLukA} \mathcal{H}}$ implies
\,${\vdash_{\GzeroLukA} \mathcal{H}}$.
\end{llemma}

\begin{proof}[\textsc{Proof}]
Any axiom of $\GthreeLukA$ is an axiom of $\GzeroLukA$.
Assuming that $(\textnormal{den}_1)$ and $(\textnormal{den}_0)$ are
admissible for $\GzeroLukA$, we establish that all the rules of $\GthreeLukA$ 
are admissible for $\GzeroLukA$.
The conclusion of the rule ${(\to\:\Rightarrow)^3}$ is obtained from 
its premise as follows:
\begin{center}
\def\ScoreOverhang{0pt}
\AxiomC{$\mathcal{G}\,|\,\Gamma, \SpV{p}_1 \Rightarrow\Delta\,|\,B \Rightarrow \SpV{p}_1, A$}
\RightLabel{$\text{(ew)}^0 \!\times\! 2$}
\UnaryInfC{$\mathcal{G}\,|\,\Gamma, \SpV{p}_1 \Rightarrow\Delta\,|\,B \Rightarrow \SpV{p}_1, A \,| \Rightarrow \SpV{p}_1 \,|\,A\to B \Rightarrow \SpV{p}_1$}
\RightLabel{$(\to\:\Rightarrow)^0$}
\UnaryInfC{$\mathcal{G}\,|\,\Gamma, \SpV{p}_1 \Rightarrow\Delta \,|\,A\to B \Rightarrow \SpV{p}_1$}
\RightLabel{$(\text{den}_1)$,}
\UnaryInfC{$\mathcal{G}\,|\,\Gamma, A\to B \Rightarrow\Delta$}
\DisplayProof
\end{center}
$\text{(ew)}^0$ being admissible for $\GzeroLukA$ 
by Lemma \ref{Lem:StructRulesAdmGzero}.
The conclusion of the rule ${(\Rightarrow\exists)^3}$ is obtained from 
its premise thus:
\begin{center} 
\def\ScoreOverhang{0pt}
\AxiomC{$\mathcal{G}\,|\,\Gamma \Rightarrow \SpV{p}_0, \Delta \,|\,\SpV{p}_0 \Rightarrow \exists x A \,|\,\SpV{p}_0 \Rightarrow \RePl{A}{x}{t}$}
\RightLabel{$(\Rightarrow\exists)^0$}
\UnaryInfC{$\mathcal{G} \,|\,\Gamma \Rightarrow \SpV{p}_0, \Delta \,|\,\SpV{p}_0 \Rightarrow \exists x A$}
\RightLabel{$(\text{den}_0)$.}
\UnaryInfC{$\mathcal{G} \,|\,\Gamma \Rightarrow \exists x A, \Delta$}
\DisplayProof
\end{center}
The rule ${(\forall\Rightarrow)^3}$ is treated similarly to
${(\Rightarrow\exists)^3}$, but with an application of $(\text{den}_1)$.
Finally, the admissibility for $\GzeroLukA$ of the rules
${(\Rightarrow\:\to)^3}$, ${(\Rightarrow\forall)^3}$, and
${(\exists\Rightarrow)^3}$ follows easily from
the admissibility of $\text{(ew)}^0$. 
\end{proof}

Lemmas \ref{Lem:DenOneAdmGzero} and \ref{Lem:DenZeroAdmGzero} below
ensure that the rules $(\textnormal{den}_1)$ and $(\textnormal{den}_0)$ 
are admissible for $\GzeroLukA$.

\begin{llemma}[admissibility of a generalization of $(\textnormal{den}_1)$ 
                for $\GzeroLukA$]
\label{Lem:DenOneAdmGzero}
Suppose that ${m\geqslant 1}$, ${n\geqslant 1}$, 
\begin{gather*}
{\mathcal{H} = \Big(\, \mathcal{G} \,\big|\, \big[ \Gamma_i, \SpV{p}_1 \Rightarrow \Delta_i \big]_{i\in 1..m} \,\big|\, \big[ \Pi_j \Rightarrow \SpV{p}_1, \Sigma_j \big]_{j\in 1..n} \,\Big)},\\
{\mathcal{H}' = \Big(\, \mathcal{G} \,\big|\, \big[ \Gamma_i, \Pi_j \Rightarrow \Delta_i, \Sigma_j \big]^{i\in 1..m}_{j\in 1..n} \,\Big)},
\end{gather*}
$\SpV{p}_1$ does not occur in $\mathcal{H}'$,
the sequent \,${C \Rightarrow \SpV{p}_1}$ occurs in $\mathcal{H}$,
and \,${\vdash_{\GzeroLukA} \mathcal{H}}$.
Then \,${\vdash_{\GzeroLukA} \mathcal{H}'}$.
\end{llemma}

\begin{proof}[\textsc{Proof}]
By Lemma \ref{Lem:GzeroProofOfAxiomForDenOne} below,
there exists a ($\GzeroLukA$-)proof $D$ of $\mathcal{H}$ in which
each leaf hypersequent $\mathcal{L}$ contains 
a sequent of the form \,${C_\mathcal{L} \Rightarrow \SpV{p}_1}$ or
\,${\Rightarrow \SpV{p}_1}$,
where $C_\mathcal{L}$ is an atomic $\RPLA$-formula.
We transform $D$ into a proof of $\mathcal{H}'$ using induction on
the height of~$D$.

1. Suppose that $\mathcal{H}$ is an axiom; i.e., ${\vDash \mathcal{H}_{at}}$.
Without loss of generality we assume that
$$\mathcal{H}_{at} = \Big(\, \mathcal{G}_{at} \,\big|\, 
  \big[ \Gamma_i, \SpV{p}_1 \Rightarrow \Delta_i \big]_{i\in 1..k} \,\big|\, 
  \big[ \Pi_j \Rightarrow \SpV{p}_1, \Sigma_j \big]_{j\in 1..l} \,\Big),$$
where \,${0 \leqslant k \leqslant m}$, \,${0 < l \leqslant n}$,\, and
the sequent \,${\Pi_1 \Rightarrow \SpV{p}_1, \Sigma_1}$\, has the form
\,${C_1 \Rightarrow \SpV{p}_1}$\, or \,${\Rightarrow \SpV{p}_1}$.
We put \,${\mathcal{H}'_{at} = (\mathcal{H}')_{at}}$.

1.1. Consider the case where ${k \neq 0}$. We have
$${\mathcal{H}'_{at} = \Big(\, \mathcal{G}_{at} \,\big|\, \big[ \Gamma_i, \Pi_j \Rightarrow \Delta_i, \Sigma_j \big]^{i\in 1..k}_{j\in 1..l} \,\Big)}.$$

We want to show that ${\vDash \mathcal{H}'_{at}}$.
Suppose otherwise; i.e.,
for some hs-interpretation $M$ and some $M$-valuation $\nu$, 
there is no true sequent in $\mathcal{G}_{at}$, and 
for all ${i\in 1..k}$ and ${j\in 1..l}$,
$$\MsV{\Delta_i}_{M,\nu} - \MsV{\Gamma_i}_{M,\nu} < \MsV{\Pi_j}_{M,\nu} - \MsV{\Sigma_j}_{M,\nu}.$$
By the density of the set $\mathbb{R}$ of all real numbers, 
there exists ${\xi \in \mathbb{R}}$ such that, 
for all ${i\in 1..k}$ and ${j\in 1..l}$,
$$\MsV{\Delta_i}_{M,\nu} - \MsV{\Gamma_i}_{M,\nu} < \xi-1 < \MsV{\Pi_j}_{M,\nu} - \MsV{\Sigma_j}_{M,\nu}.$$
In particular, ${\xi < \MsV{\Pi_1}_{M,\nu} - \MsV{\Sigma_1}_{M,\nu} + 1 =
                 \MsV{\Pi_1}_{M,\nu} + 1 \leqslant 1}$.

Define an hs-interpretation $M_1$ to be like $M$, but set
${|\SpV{p}_1|_{M_1} = \xi}$.
Since $\SpV{p}_1$ does not occur in $\mathcal{G}_{at}$, 
$\Gamma_i$, $\Delta_i$ (${i\in 1..k}$), $\Pi_j$, $\Sigma_j$ (${j\in 1..l}$), 
we see that no sequent in $\mathcal{H}_{at}$ is true 
under the hs-interpretation $M_1$ and $M_1$-valuation $\nu$.
Hence ${\nvDash \mathcal{H}_{at}}$, a contradiction. 

Therefore ${\vDash \mathcal{H}'_{at}}$, and so $\mathcal{H}'$ is an axiom.

1.2. Now consider the case where ${k = 0}$. Then
\,$\mathcal{H}_{at} = \Big( \mathcal{G}_{at} \,\big|\, 
    \big[ \Pi_j \Rightarrow \SpV{p}_1, \Sigma_j \big]_{j\in 1..l} \Big)$
and \,${\mathcal{H}'_{at} = \mathcal{G}_{at}}$.
Since $\SpV{p}_1$ does not occur in $\mathcal{G}_{at}$, $\Pi_j$, $\Sigma_j$
($j\in 1..l$), and 
hs-in\-ter\-pre\-ta\-ti\-ons can take $\SpV{p}_1$ to negative real numbers whose 
absolute values are arbitrarily large, 
we conclude that 
${\vDash \mathcal{H}_{at}}$ implies ${\vDash \mathcal{G}_{at}}$.
Thus ${\vDash \mathcal{H}'_{at}}$ and $\mathcal{H}'$ is an axiom.

2. Suppose that the root hypersequent $\mathcal{H}$ in $D$ is the conclusion of
an application $R$ of a rule $\mathcal{R}$, and 
$\mathcal{S}$ is the principal sequent occurrence in~$R$.

2.1. If $\mathcal{S}$ is in the distinguished occurrence of $\mathcal{G}$ 
in $\mathcal{H}$, then we apply the induction hypothesis to the proof of each 
premise of $R$, and next we get a proof of $\mathcal{H}'$ by $\mathcal{R}$.

2.2. Now suppose that $\mathcal{S}$ is not in the distinguished occurrence of 
$\mathcal{G}$ in $\mathcal{H}$, and for definiteness assume that
$\mathcal{S}$ is the distinguished occurrence of 
\,${\Gamma_1, \SpV{p}_1 \Rightarrow \Delta_1}$ in $\mathcal{H}$.

2.2.1. If $\mathcal{R}$ is the rule ${(\to\:\Rightarrow)^0}$, then
${\Gamma_1 = (\Gamma_1',\, A \to B)}$ for some $\Gamma_1'$,
and the proof $D$ has the form:
\begin{center}
\def\ScoreOverhang{0pt}
\AxiomC{$D_1$} \noLine 
\UnaryInfC{$\mathcal{H} \,|\, \Gamma_1', \SpV{p}_1 \Rightarrow \Delta_1 \,|\, \Gamma_1', B, \SpV{p}_1 \Rightarrow A, \Delta_1$}
\RightLabel{${(\to\:\Rightarrow)^0}$.} 
\UnaryInfC{$\mathcal{H}$}
\DisplayProof
\end{center}
By the induction hypothesis, we transform $D_1$ into a proof of
$$\mathcal{H}' \,\big|\, \big[ \Gamma_1', \Pi_j \Rightarrow \Delta_1, \Sigma_j \big]_{j\in 1..n}
  \,\big|\, \big[ \Gamma_1', B, \Pi_j \Rightarrow A, \Delta_1, \Sigma_j \big]_{j\in 1..n},$$
whence we obtain a proof for $\mathcal{H}'$
by $n$ applications of ${(\to\:\Rightarrow)^0}$.

2.2.2. The rules ${(\forall\Rightarrow)^0}$ and ${(\Rightarrow\exists)^0}$
are treated as $\mathcal{R}$ similarly to the rule ${(\to\:\Rightarrow)^0}$, 
see item~2.2.1.

2.2.3. If $\mathcal{R}$ is ${(\Rightarrow\:\to)^0}$, then
${\Delta_1 = (A \to B,\, \Delta_1')}$ for some $\Delta_1'$,
and the proof $D$ looks like this:
\begin{center}
\def\ScoreOverhang{0pt}
\AxiomC{$D_1$} \noLine 
\UnaryInfC{$\mathcal{H} \,|\, \Gamma_1, \SpV{p}_1 \Rightarrow \Delta_1'$;}
\AxiomC{$D_2$} \noLine 
\UnaryInfC{$\mathcal{H} \,|\, \Gamma_1, A, \SpV{p}_1 \Rightarrow B, \Delta_1'$}
\RightLabel{${(\Rightarrow\:\to)^0}$.}
\BinaryInfC{$\mathcal{H}$}
\DisplayProof
\end{center}
By the induction hypothesis applied to the proofs $D_1$ and $D_2$,
we construct proofs of
$$\mathcal{H}' \,\big|\, \big[ \Gamma_1, \Pi_j \Rightarrow \Delta_1', \Sigma_j \big]_{j\in 1..n}
  \quad \text{and} \quad
  \mathcal{H}' \,\big|\, \big[ \Gamma_1, A, \Pi_j \Rightarrow B, \Delta_1', \Sigma_j \big]_{j\in 1..n},$$
respectively; 
whence we get a proof of $\mathcal{H}'$ 
by Lemma \ref{Lem:GzeroAdmGenToRightRule} below.

2.2.4. If $\mathcal{R}$ is ${(\Rightarrow\forall)^0}$,
then ${\Delta_1 = (\forall x A,\, \Delta_1')}$ for some $\Delta_1'$,
and the proof $D$ has the form:
\begin{center}
\def\ScoreOverhang{0pt}
\AxiomC{$D_1$} \noLine 
\UnaryInfC{$\mathcal{H} \,|\, \Gamma_1, \SpV{p}_1 \Rightarrow \RePl{A}{x}{a}, \Delta_1'$}
\RightLabel{${(\Rightarrow\forall)^0}$,}
\UnaryInfC{$\mathcal{H}$}
\DisplayProof
\end{center}
where $a$ does not occur in $\mathcal{H}$ 
(and hence, $a$ does not occur in $\mathcal{H}'$).
Using the induction hypothesis, we transform $D_1$ into a proof of
$$\mathcal{H}' \,\big|\, \big[ \Gamma_1, \Pi_j \Rightarrow \RePl{A}{x}{a}, \Delta_1', \Sigma_j \big]_{j\in 1..n},$$
whence we obtain a proof of $\mathcal{H}'$ 
by Lemma~\ref{Lem:GzeroAdmGenForallRight}.

2.2.5. The rule ${(\exists\Rightarrow)^0}$ is treated 
similarly to the rule ${(\Rightarrow\forall)^0}$ in item 2.2.4,
using Lemma~\ref{Lem:GzeroAdmGenExistsLeft}.
\end{proof}

\begin{llemma}  \label{Lem:GzeroProofOfAxiomForDenOne}
Suppose that \,${\mathcal{H} = (\mathcal{G} \,|\, C \Rightarrow \SpV{p}_1)}$
is an axiom of $\GzeroLukA$.
Then a $\GzeroLukA$-proof of $\mathcal{H}$ can be constructed in which 
each leaf hypersequent $\mathcal{L}$ contains  
a sequent of the form \,${C_\mathcal{L} \Rightarrow \SpV{p}_1}$ or
\,${\Rightarrow \SpV{p}_1}$,
where $C_\mathcal{L}$ is an atomic $\RPLA$-formula.
\end{llemma}

\begin{proof}[\textsc{Proof}]
The $\RPLA$-formula $C$ has the form 
$$Q_1 x_1 \ldots Q_n x_n C'  \quad \text{or} \quad  
  Q_1 x_1 \ldots Q_n x_n (A \to B),$$ 
where 
$Q_1,\ldots,Q_n$ are quantifiers and $C'$ is an atomic $\RPLA$-formula.
The desired proof can be obtained from $\mathcal{H}$ 
by $n$ backward applications of 
the rules ${(Q_1\Rightarrow)^0}$, \ldots, ${(Q_n\Rightarrow)^0}$
and if ${C = Q_1 x_1 \ldots Q_n x_n (A \to B)}$, by one more
backward application of the rule ${(\to\:\Rightarrow)^0}$.
\end{proof}

\begin{llemma}  \label{Lem:GzeroAdmGenToRightRule}
Suppose that ${n \geqslant 1}$,
$${\mathcal{H}_n' = \Big(\, \mathcal{G} \,\big|\, \big[ \Gamma_i \Rightarrow \Delta_i \big]_{i\in 1..n} \,\Big)}, 
  \quad  
  {\mathcal{H}_n'' = \Big(\, \mathcal{G} \,\big|\, \big[ \Gamma_i, A \Rightarrow B, \Delta_i \big]_{i\in 1..n} \,\Big)},$$
${\vdash_{\GzeroLukA} \mathcal{H}_n'}$,
\,${\vdash_{\GzeroLukA} \mathcal{H}_n''}$, and
${\big[ \Gamma_i \Rightarrow A \to B, \Delta_i \big]_{i\in 1..n} \!\subseteq \mathcal{G}}$.
Then \,${\vdash_{\GzeroLukA} \mathcal{G}}$.
\end{llemma}

\begin{proof}[\textsc{Proof}]
We proceed by induction on $n$. 
The case ${n=1}$ is trivial. 

Now suppose that ${n \geqslant 2}$.\,
By Lemma \ref{Lem:GzeroExchange} below,
from \,${\vdash_{\GzeroLukA} \mathcal{H}_n'}$ and
\,${\vdash_{\GzeroLukA} \mathcal{H}_n''}$\, it follows that the hypersequent
$${\mathcal{H}_n = \Big(\, \mathcal{G} \,\big|\, \big[ \Gamma_i \Rightarrow \Delta_i \big]_{i\in 1..(n-1)} \,\big|\, \Gamma_n,A \Rightarrow B,\Delta_n \,\Big)}$$
is $\GzeroLukA$-provable.
Applying the rule ${(\Rightarrow\:\to)^0}$ 
to $\mathcal{H}_n'$ and $\mathcal{H}_n$ gives 
$${\mathcal{H}_{n-1}' = \Big(\, \mathcal{G} \,\big|\, \big[ \Gamma_i \Rightarrow \Delta_i \big]_{i\in 1..(n-1)} \,\Big)}.$$
Likewise we obtain the $\GzeroLukA$-provable hypersequent
$${\mathcal{H}_{n-1}'' = \Big(\, \mathcal{G} \,\big|\, \big[ \Gamma_i,A \Rightarrow B,\Delta_i \big]_{i\in 1..(n-1)} \,\Big)}.$$
Finally, by applying the induction hypothesis to $\mathcal{H}_{n-1}'$ and
$\mathcal{H}_{n-1}''$, we get \,${\vdash_{\GzeroLukA} \mathcal{G}}$.
\end{proof}

\begin{llemma}  \label{Lem:GzeroExchange}
Suppose that ${n \geqslant 2}$, 
$$\mathcal{H}' \!=\! \Big( \mathcal{G} \,\big|\, \big[ \Gamma_i, \Pi' \Rightarrow \Sigma', \Delta_i \big]_{i\in 1..n} \Big),
  \;\:
  \mathcal{H}'' \!=\! \Big( \mathcal{G} \,\big|\, \big[ \Gamma_i, \Pi'' \Rightarrow \Sigma'', \Delta_i \big]_{i\in 1..n} \Big),$$
${\vdash_{\GzeroLukA} \mathcal{H}'}$, and
\,${\vdash_{\GzeroLukA} \mathcal{H}''}$.
Then 
$${\vdash_{\GzeroLukA} \Big(\, \mathcal{G} \,\big|\, \big[ \Gamma_i, \Pi' \Rightarrow \Sigma', \Delta_i \big]_{i\in 1..(n-1)} \,\big|\, \Gamma_n, \Pi'' \Rightarrow \Sigma'', \Delta_n \,\Big)}.$$
\end{llemma}

\begin{proof}[\textsc{Proof}]
For each ${k \in 1..n}$, we put
$${\mathcal{H}_k = \Big(\, \mathcal{G} \,\big|\, \big[ \Gamma_i, \Pi' \Rightarrow \Sigma', \Delta_i \big]_{i\in 1..(n-1)} \,\big|\, \big[ \Gamma_i, \Pi'' \Rightarrow \Sigma'', \Delta_i \big]_{i\in k..n} \,\Big)}.$$
We can get $\mathcal{H}_1$ from $\mathcal{H}''$ 
by the rule $\textnormal{(ew)}^0$.
For each ${k \in 1..(n-1)}$,
Figure \ref{fig:GenRulesAdmGzeroObtainHkPlusOne} shows how to obtain
$\mathcal{H}_{k+1}$ from $\mathcal{H}'$ and $\mathcal{H}_k$ 
using the rules $\textnormal{(ew)}^0$, $\textnormal{(ec)}^0$, 
$\textnormal{(split)}^0$, and $\textnormal{(mix)}^0$.
\begin{figure}[!tb]
\begin{center} \small
\def\ScoreOverhang{0pt}
\AxiomC{$\mathcal{H}'$} 
\LeftLabel{$\genfrac{}{}{0pt}{}{\displaystyle \textnormal{(ew)}^0 \!\times}{\displaystyle (n-k)}$}
\UnaryInfC{$\mathcal{G} \,\big|\, \big[ \Gamma_i, \Pi' \Rightarrow \Sigma', \Delta_i \big]_{i\in 1..n} \,\big|\, \big[ \Gamma_i, \Pi'' \Rightarrow \Sigma'', \Delta_i \big]_{i\in (k+1)..n}$;} 
\AxiomC{$\mathcal{H}_k$} 
\RightLabel{$\textnormal{(mix)}^0$}
\BinaryInfC{$\genfrac{}{}{0pt}{}{\displaystyle 
    \mathcal{G} \,\big|\, \big[ \Gamma_i, \Pi' \Rightarrow \Sigma', \Delta_i \big]_{i\in 1..(n-1)} \,\big|\, \big[ \Gamma_i, \Pi'' \Rightarrow \Sigma'', \Delta_i \big]_{i\in (k+1)..n}
  }{\displaystyle 
    \big|\, \Gamma_n,\Pi', \Gamma_k,\Pi'' \Rightarrow \Sigma',\Delta_n, \Sigma'',\Delta_k 
  }$}
\RightLabel{$\textnormal{(split)}^0$}
\UnaryInfC{$\genfrac{}{}{0pt}{}{\displaystyle 
    \mathcal{G} \,\big|\, \big[ \Gamma_i , \Pi' \Rightarrow \Sigma', \Delta_i \big]_{i\in 1..(n-1)} \,\big|\, \big[ \Gamma_i, \Pi'' \Rightarrow \Sigma'', \Delta_i \big]_{i\in (k+1)..n}
  }{\displaystyle 
    \big|\, \Gamma_k, \Pi' \Rightarrow \Sigma', \Delta_k \,\big|\, \Gamma_n, \Pi'' \Rightarrow \Sigma'', \Delta_n 
  }$}
\RightLabel{$\textnormal{(ec)}^0 \!\!\times\! 2$}
\UnaryInfC{$\mathcal{G} \,\big|\, \big[ \Gamma_i, \Pi' \Rightarrow \Sigma', \Delta_i \big]_{i\in 1..(n-1)} \,\big|\, \big[ \Gamma_i, \Pi'' \Rightarrow \Sigma'', \Delta_i \big]_{i\in (k+1)..n}$} 
\DisplayProof
\caption{Obtaining the bottom hypersequent~$\mathcal{H}_{k+1}$.} 
\label{fig:GenRulesAdmGzeroObtainHkPlusOne}
\end{center}
\end{figure}
Recall that these four rules are admissible for $\GzeroLukA$
by Lemma \ref{Lem:StructRulesAdmGzero}.
So \,${\vdash_{\GzeroLukA} \mathcal{H}_n}$ as required.
\end{proof}

\begin{llemma}  \label{Lem:GzeroAdmGenForallRight}
Suppose that ${n \geqslant 1}$,
$${\vdash_{\GzeroLukA} \Big(\, \mathcal{G} \,\big|\, \big[ \Gamma_i \Rightarrow \RePl{A}{x}{a}, \Delta_i \big]_{i\in 1..n} \,\Big)},$$
${\big[ \Gamma_i \Rightarrow \forall x A, \Delta_i \big]_{i\in 1..n} \subseteq \mathcal{G}}$,\,
and the parameter $a$ does not occur in $\mathcal{G}$.
Then \,${\vdash_{\GzeroLukA} \mathcal{G}}$.
\end{llemma}

\begin{proof}[\textsc{Proof}]
We can obtain $\mathcal{G}$ from
\,${\mathcal{G} \,\big|\, \big[ \Gamma_i \Rightarrow \RePl{A}{x}{a_i}, \Delta_i \big]_{i\in 1..n}}$
by $n$ applications of the rule ${(\Rightarrow\forall)^0}$,
provided the parameters $a_1,\ldots,a_n$ are distinct and 
none of them occurs in $\mathcal{G}$.

Therefore, it suffices to prove the following claim for every ${n \geqslant 1}$:
\textsl{suppose that
\,${\mathcal{H}(a) = \Big(\, \mathcal{G}_0 \,\big|\, \big[ \Gamma_i \Rightarrow \RePl{A}{x}{a}, \Delta_i \big]_{i\in 1..n} \,\Big)}$,
\,${\vdash_{\GzeroLukA} \mathcal{H}(a)}$, and
the parameters $a,a_1,\ldots,a_n$ are distinct and 
none of them occurs in
\,$\mathcal{G}_0$, $A$, $\Gamma_i$, $\Delta_i$ \textnormal{(}${i\in 1..n}$\textnormal{)};\,
then
\,${\vdash_{\GzeroLukA} \Big(\, \mathcal{G}_0 \,\big|\, \big[ \Gamma_i \Rightarrow \RePl{A}{x}{a_i}, \Delta_i \big]_{i\in 1..n} \,\Big)}$}.

We use induction on~$n$.
In the case ${n=1}$, the claim is obvious.

Suppose that ${n \geqslant 2}$.
Clearly, \,${\vdash_{\GzeroLukA} \mathcal{H}(a)}$ implies
\,${\vdash_{\GzeroLukA} \mathcal{H}(a_n)}$.
By Lemma \ref{Lem:GzeroExchange}, from \,${\vdash_{\GzeroLukA} \mathcal{H}(a)}$
and \,${\vdash_{\GzeroLukA} \mathcal{H}(a_n)}$ it follows that
$$\vdash_{\GzeroLukA} \Big(\, \mathcal{G}_0 \,\big|\, \big[ \Gamma_i \Rightarrow \RePl{A}{x}{a}, \Delta_i \big]_{i\in 1..(n-1)} \,\big|\, \Gamma_n \Rightarrow \RePl{A}{x}{a_n}, \Delta_n \,\Big),$$
whence by the induction hypothesis, we get what is required.
\end{proof}

\begin{llemma}  \label{Lem:GzeroAdmGenExistsLeft}
Suppose that ${n \geqslant 1}$,
$${\vdash_{\GzeroLukA} \Big(\, \mathcal{G} \,\big|\, \big[ \Gamma_i, \RePl{A}{x}{a} \Rightarrow \Delta_i \big]_{i\in 1..n} \,\Big)},$$
${\big[ \Gamma_i, \exists x A \Rightarrow \Delta_i \big]_{i\in 1..n} \subseteq \mathcal{G}}$,\,
and the parameter $a$ does not occur in $\mathcal{G}$.
Then \,${\vdash_{\GzeroLukA} \mathcal{G}}$.
\end{llemma}

\begin{proof}[\textsc{Proof}\nopunct]
is similar to the proof of Lemma~\ref{Lem:GzeroAdmGenForallRight}.
\end{proof}

For a finite multiset $\Delta$, by $\#(\Delta)$ we denote the number of 
its elements, taking their multiplicities into account.

\begin{llemma}[admissibility of a generalization of $(\textnormal{den}_0)$ 
                for $\GzeroLukA$]
\label{Lem:DenZeroAdmGzero}
Suppose that ${m\geqslant 1}$, ${n\geqslant 1}$, 
\begin{gather*}
{\mathcal{H} = \Big(\, \mathcal{G} \,\big|\, \big[ \Gamma_i, \SpV{p}_0 \Rightarrow \Delta_i \big]_{i\in 1..m} \,\big|\, \big[ \Pi_j \Rightarrow \SpV{p}_0, \Sigma_j \big]_{j\in 1..n} \,\Big)},\\
{\mathcal{H}' = \Big(\, \mathcal{G} \,\big|\, \big[ \Gamma_i, \Pi_j \Rightarrow \Delta_i, \Sigma_j \big]^{i\in 1..m}_{j\in 1..n} \,\Big)},
\end{gather*}
$\SpV{p}_0$ does not occur in $\mathcal{H}'$,
the sequent \,${\SpV{p}_0 \Rightarrow C}$ occurs in $\mathcal{H}$,
and \,${\vdash_{\GzeroLukA} \mathcal{H}}$.
Then \,${\vdash_{\GzeroLukA} \mathcal{H}'}$.
\end{llemma}

\begin{proof}[\textsc{Proof}]
Using Lemma \ref{Lem:GzeroProofOfAxiomForDenZero} below,
we find a ($\GzeroLukA$-)proof $D$ of $\mathcal{H}$ in which
each leaf hypersequent $\mathcal{L}$ contains an atomic sequent of the form
\,${\Gamma_\mathcal{L}, \SpV{p}_0 \Rightarrow \Delta_\mathcal{L}}$,
where \,${\#(\Delta_\mathcal{L}) \leqslant 1}$ and
no semipropositional variable occurs in $\Gamma_\mathcal{L}$ or 
$\Delta_\mathcal{L}$.
We show that \,${\vdash_{\GzeroLukA} \mathcal{H}'}$
by induction on the height of~$D$.

1. Suppose that $\mathcal{H}$ is an axiom; i.e., ${\vDash \mathcal{H}_{at}}$.
We can harmlessly assume that
$$\mathcal{H}_{at} = \Big(\, \mathcal{G}_{at} \,\big|\, 
  \big[ \Gamma_i, \SpV{p}_0 \Rightarrow \Delta_i \big]_{i\in 1..k} \,\big|\, 
  \big[ \Pi_j \Rightarrow \SpV{p}_0, \Sigma_j \big]_{j\in 1..l} \,\Big),$$
where \,${0 < k \leqslant m}$, \,${0 \leqslant l \leqslant n}$\, and
\,${\#(\Delta_1) \leqslant 1}$.
We put \,${\mathcal{H}'_{at} = (\mathcal{H}')_{at}}$.

1.1. Consider the case where ${l \neq 0}$. We have
$${\mathcal{H}'_{at} = \Big(\, \mathcal{G}_{at} \,\big|\, \big[ \Gamma_i, \Pi_j \Rightarrow \Delta_i, \Sigma_j \big]^{i\in 1..k}_{j\in 1..l} \,\Big)}.$$

Suppose for a contradiction that \,${\nvDash \mathcal{H}'_{at}}$; i.e., 
for some hs-in\-ter\-pre\-ta\-tion $M$ and some $M$-valuation $\nu$, 
there is no true sequent in $\mathcal{G}_{at}$, and
for all ${i\in 1..k}$ and ${j\in 1..l}$, 
$$\MsV{\Delta_i}_{M,\nu} - \MsV{\Gamma_i}_{M,\nu} < \MsV{\Pi_j}_{M,\nu} - \MsV{\Sigma_j}_{M,\nu}.$$
Hence, for some real number $\xi$ and for all ${i\in 1..k}$ and ${j\in 1..l}$,
$$\MsV{\Delta_i}_{M,\nu} - \MsV{\Gamma_i}_{M,\nu} < \xi-1 < \MsV{\Pi_j}_{M,\nu} - \MsV{\Sigma_j}_{M,\nu}.$$
In particular, ${\xi > \MsV{\Delta_1}_{M,\nu} - \MsV{\Gamma_1}_{M,\nu} + 1 \geqslant
                 \MsV{\Delta_1}_{M,\nu} + 1 \geqslant 0}$.

Define an hs-interpretation $M_0$ to be the same as $M$ except that
${|\SpV{p}_0|_{M_0} = \xi}$.
Since $\SpV{p}_0$ does not occur in $\mathcal{G}_{at}$, 
$\Gamma_i$, $\Delta_i$ (${i\in 1..k}$), $\Pi_j$, $\Sigma_j$ (${j\in 1..l}$), 
it follows that $\mathcal{H}_{at}$ has no true sequent
under the hs-in\-ter\-pre\-ta\-tion $M_0$ and $M_0$-valuation $\nu$.
So ${\nvDash \mathcal{H}_{at}}$, a contradiction.

Thus ${\vDash \mathcal{H}'_{at}}$ and $\mathcal{H}'$ is an axiom.

1.2. Now consider the case where ${l = 0}$. Then
\,$\mathcal{H}_{at} \!=\! \Big( \mathcal{G}_{at} \,\big|\, 
    \big[ \Gamma_i, \SpV{p}_0 \Rightarrow \Delta_i \big]_{i\in 1..k} \Big)$
and \,${\mathcal{H}'_{at} = \mathcal{G}_{at}}$.
Since $\SpV{p}_0$ does not occur in $\mathcal{G}_{at}$, $\Gamma_i$, $\Delta_i$ 
($i\in 1..k$), and $\SpV{p}_0$ can assume arbitrarily large values
under hs-in\-ter\-pre\-ta\-ti\-ons, we see that
${\vDash \mathcal{H}_{at}}$ implies ${\vDash \mathcal{G}_{at}}$.
So ${\vDash \mathcal{H}'_{at}}$ and $\mathcal{H}'$ is an axiom.

2. It remains to consider the case where the root hypersequent $\mathcal{H}$ 
in $D$ is the conclusion of a rule application.
But the argument for this case can be obtained from
item 2 of the proof of Lemma \ref{Lem:DenOneAdmGzero} 
by replacing $\SpV{p}_1$ with~$\SpV{p}_0$.
\end{proof}

\begin{llemma} \label{Lem:GzeroProofOfAxiomForDenZero}
Suppose that 
\,${\mathcal{H} = (\mathcal{G} \,|\, \Gamma, \SpV{p}_0 \Rightarrow \Delta)}$
is an axiom of $\GzeroLukA$, $\#(\Delta) \leqslant 1$, and
no semipropositional variable occurs in $\Gamma$ or $\Delta$.
Then a $\GzeroLukA$-proof of $\mathcal{H}$ can be constructed in which
each leaf hypersequent $\mathcal{L}$ contains an atomic sequent of the form
\,${\Gamma_\mathcal{L}, \SpV{p}_0 \Rightarrow \Delta_\mathcal{L}}$,
where \,${\#(\Delta_\mathcal{L}) \leqslant 1}$ and
no semipropositional variable occurs in $\Gamma_\mathcal{L}$ or 
$\Delta_\mathcal{L}$.
\end{llemma}

\begin{proof}[\textsc{Proof}\nopunct]
is by induction on the number of logical symbol occurrences in
the sequent ${S = (\Gamma, \SpV{p}_0 \Rightarrow \Delta)}$.
If $S$ is atomic, then $\mathcal{H}$ is the desired proof.
Otherwise, $S$ has one of the forms given in items 1--4 below. 

1. Suppose that ${S = (\Gamma', A\to B, \SpV{p}_0 \Rightarrow \Delta)}$.
By applying the rule ${(\to\:\Rightarrow)^0}$ backward to 
the distinguished occurrence of ${A\to B}$ in $\mathcal{H}$,
we get the axiom 
${\mathcal{H}_1 = (\mathcal{G} \,|\, S \,|\, \Gamma', \SpV{p}_0 \Rightarrow \Delta \,|\, \Gamma', B, \SpV{p}_0 \Rightarrow A, \Delta)}$.
By the induction hypothesis applied to $\mathcal{H}_1$ with
${S = (\Gamma', \SpV{p}_0 \Rightarrow \Delta)}$,
we obtain the desired proof of~$\mathcal{H}$.

2. Suppose that ${S = (\Gamma, \SpV{p}_0 \Rightarrow A\to B)}$.
Applying the rule ${(\Rightarrow\:\to)^0}$ backward to
the distinguished occurrence of ${A\to B}$ in $\mathcal{H}$
yields the axioms 
${(\mathcal{G} \,|\, S \,|\, \Gamma, \SpV{p}_0 \Rightarrow)}$ and
${(\mathcal{G} \,|\, S \,|\, \Gamma, A, \SpV{p}_0 \Rightarrow B)}$,
to each of which the induction hypothesis applies.

3. Suppose that ${S = (\Gamma, \SpV{p}_0 \Rightarrow Q x A)}$, 
where $Q$ is a quantifier.
We apply the rule ${(\Rightarrow\forall)^0}$ or
${(\Rightarrow\exists)^0}$ backward to 
the distinguished occurrence of $Q x A$ in $\mathcal{H}$ 
with a new parameter $a$ as the proper parameter or proper term, respectively. 
Thus we get the axiom 
${(\mathcal{G} \,|\, S \,|\, \Gamma, \SpV{p}_0 \Rightarrow \RePl{A}{x}{a})}$
and then use the induction hypothesis.

4. The case where ${S = (\Gamma', Q x A, \SpV{p}_0 \Rightarrow \Delta)}$, 
with $Q$ being a quantifier, is treated similarly to case~3.
\end{proof}

\begin{remark} \label{rem:DensityAdmGzero}
The proofs of Lemmas \ref{Lem:DenOneAdmGzero} and \ref{Lem:DenZeroAdmGzero}
can be easily combined to establish the admissibility of 
the nonstandard density rule (given in Remark \ref{rem:DensitySound} 
on p.~\pageref{rem:DensitySound}) for $\GzeroLukA$ (with the notion of
a hypersequent expanded as mentioned in Remark~\ref{rem:DensitySound}).\footnote{
 See also Section \ref{secApp:DenAdmGzero} (of the appendix)
 on p.~\pageref{secApp:DenAdmGzero}.
}
\end{remark}

\begin{ltheorem}[equivalence of $\GzeroLukA$, $\GonezLukA$, and $\GthreeLukA$]
\label{Th:GzeroGonezGthreeEquiv}
\ ${\vdash_{\GzeroLukA} \mathcal{H}}$ iff \:${\vdash_{\GonezLukA} \mathcal{H}}$ 
iff \:${\vdash_{\GthreeLukA} \mathcal{H}}$.
\end{ltheorem}

\begin{proof}[\textsc{Proof}]
\,${\vdash_{\GzeroLukA} \mathcal{H}}$ implies
\,${\vdash_{\GonezLukA} \mathcal{H}}$ by Theorem \ref{Th:GzeroImpGonez}.
If \,${\vdash_{\GonezLukA} \mathcal{H}}$, then
\,${\vdash_{\GthreeLukA} \mathcal{H}}$ by Theorem \ref{Th:GonezImpGthree}.
By Lemmas \ref{Lem:IfDeniAdmGzero}, \ref{Lem:DenOneAdmGzero}, and
\ref{Lem:DenZeroAdmGzero}, from \,${\vdash_{\GthreeLukA} \mathcal{H}}$
it follows that \,${\vdash_{\GzeroLukA} \mathcal{H}}$.
\end{proof}

\begin{ltheorem}  \label{Th:GzeroGonezGthreeConsExtGone}
For each ${i=0,\hat{1},3}$, the calculus $\GiLukA$ is
a conservative extension of \,$\GoneLukA$; i.e., 
for any $\RPLA^1$-hypersequent $\mathcal{H}$,
\,${\vdash_{\GiLukA} \mathcal{H}}$ iff \:${\vdash_{\GoneLukA} \mathcal{H}}$.
\end{ltheorem}

\begin{proof}[\textsc{Proof}]
Follows from Remark \ref{rem:GonezGone} on p.~\pageref{rem:GonezGone}
and Theorem \ref{Th:GzeroGonezGthreeEquiv}.
\end{proof}

\begin{ltheorem}  \label{Th:GzeroGoneGonezGthreeConsExtGLukA}
For each ${i=0,1,\hat{1},3}$, the calculus $\GiLukA$ is
a conservative extension of \,$\GLukA$; i.e., 
for any $\LukA$-hypersequent $\mathcal{H}$,
\,${\vdash_{\GiLukA} \mathcal{H}}$ iff \:${\vdash_{\GLukA} \mathcal{H}}$.
\end{ltheorem}

\begin{proof}[\textsc{Proof}]
Immediate from Theorems \ref{Th:GzeroConsExtGLukA} and 
\ref{Th:GzeroGonezGthreeConsExtGone}.
\end{proof}

For $\RPLA$- and $\LukA$-formulas, the preceding three theorems yield
the next result.

\begin{lcorollary}  \label{Cor:GiConsExtForFormulas}
\mbox{ }\\
\textnormal{(a)} For each ${i=0,1,\hat{1}}$ and any $\RPLA$-formula $A$,
\,${\vdash_{\GiLukA} A}$ iff \:${\vdash_{\GthreeLukA} A}$.\\
\textnormal{(b)} For each ${i=0,1,\hat{1},3}$ and any $\LukA$-formula $A$,
\,${\vdash_{\GiLukA} A}$ iff \:${\vdash_{\GLukA} A}$.
\end{lcorollary}

\section{Some relationships between $\GzeroLukA$, $\GthreeLukA$, and\\
Hilbert-type calculi for $\RPLA$ and $\HLukA$}
\label{sec:HilRelat}

In this section we compare the hypersequent calculi $\GzeroLukA$ and
$\GthreeLukA$ with a Hilbert-type calculus $\HRPLA$ for the logic $\RPLA$
(cf. \cite{Hajek1998}), as well as with a Hilbert-type calculus $\HLukA$ 
for the logic $\LukA$ (cf. \cite{Hajek1998}).
First we formulate $\HRPLA$ and $\LukA$.

The axiom schemes of $\HRPLA$ are:

(\L1)  $A \to (B \to A)$;

(\L2)  $(A \to B) \to ((B \to C) \to (A \to C))$;

(\L3)  $(\neg A \to \neg B) \to (B \to A)$,
       \ where $\neg C$ is short for $(C \to \bar{0})$;

(\L4)  $((A \to B) \to B) \to ((B \to A) \to A)$;

(tc1)  $(\bar{r}_1 \to \bar{r}_2) \to \bar{r}$,
       \ where \,$r = \min(1-r_1+r_2, \,1)$;

(tc2)  $\bar{r} \to (\bar{r}_1 \to \bar{r}_2)$,
       \ where \,$r = \min(1-r_1+r_2, \,1)$;

($\forall$1)  $\forall x A \to \RePl{A}{x}{t}$,
              \ where $t$ is a term (not necessarily closed)
              free for $x$ in $A$;

($\forall$2)  $\forall x (A \to B) \to (A \to \forall x B)$,
              \ where $x$ does not occur free in $A$;

($\exists$1)  $\RePl{A}{x}{t} \to \exists x A$,
              \ where $t$ is a term (not necessarily closed)
              free for $x$ in $A$;

($\exists$2)  $\forall x (A \to B) \to (\exists x A \to B)$,
              \ where $x$ does not occur free in $B$.

The inference rules of $\HRPLA$ are:
$$\dfrac{A;  \quad  A \to B}{B}~\text{(mp)},  \quad\quad\quad  
  \dfrac{A}{\forall x A}~\text{(gen)}.$$

To obtain $\HLukA$ from $\HRPLA$, we require 
$A$, $B$, and $C$ to be $\LukA$-for\-mu\-las in the formulation of $\HRPLA$ 
and remove the axiom schemes (tc1) and (tc2) from it.

As hypersequent counterparts of the rules (mp) of $\HRPLA$ and $\HLukA$,
we consider the following cut rules (cf., e.g., \cite[Section 4.2]{MOG2009}):
$$\dfrac{\mathcal{G}\,|\, \Gamma_1 \Rightarrow C,\Delta_1;  \quad  \mathcal{G}\,|\, \Gamma_2,C \Rightarrow \Delta_2}
   {\mathcal{G}\,|\, \Gamma_1,\Gamma_2\Rightarrow\Delta_1,\Delta_2} 
   \text{\small \arraycolsep=0.5mm 
     $\left\{ 
        \begin{array}{ll} 
          \text{(cut)}, & \text{where $C$ is an $\RPLA$-formula,}\\
          \text{(\L cut)}, & \text{where $C$ is an $\LukA$-formula.}
      \end{array} \right.$
   }
$$

\begin{lproposition}  \label{Pr:CutNotAdm}
The rule \textnormal{(\L cut)} 
\textnormal{(}and hence \textnormal{(cut)}\textnormal{)} 
is admissible neither for $\GzeroLukA$ nor for $\GthreeLukA$.
\end{lproposition}

\begin{proof}[\textsc{Proof}]
In \cite[p.~268]{Metcalfe2011}, 
for some $\LukA$-sentence $A$, it is shown that \,${\nvdash_{\GLukA} A}$,\, and
a proof is constructed of the form
\begin{center}
\def\ScoreOverhang{0pt}
\AxiomC{$D_1$} \noLine 
\UnaryInfC{$\mathcal{H}_1$;}
\AxiomC{$D_2$} \noLine 
\UnaryInfC{$\mathcal{H}_2$}
\RightLabel{(\L cut),}
\BinaryInfC{$\Rightarrow A$}
\DisplayProof
\end{center}
where $D_1$ and $D_2$ are $\GLukA$-proofs.
By Theorem \ref{Th:GzeroGoneGonezGthreeConsExtGLukA}, for ${i=0,3}$, we get
\,${\vdash_{\GiLukA} \mathcal{H}_1}$ and \,${\vdash_{\GiLukA} \mathcal{H}_2}$,
whence \,${\vdash_{\GiLukA+\textnormal{(\L cut)}} A}$.
But by Corollary \ref{Cor:GiConsExtForFormulas},
we have \,${\nvdash_{\GiLukA} A}$\, for ${i=0,3}$.
\end{proof}

In the rest of this section we establish the following two theorems.

\begin{ltheorem}  \label{Th:GzeroGthreeHRPLA}
For any $\RPLA$-sentence $A$,
\begin{center}
${\vdash_{\HRPLA} A}$ implies \,${\vdash_{\GthreeLukA+\textnormal{(cut)}} A}$,
which in turn implies \,${\vdash_{\GzeroLukA+\textnormal{(cut)}} A}$.
\end{center}
\end{ltheorem}

\begin{ltheorem}  \label{Th:GzeroGthreeHLukA}
For any $\LukA$-sentence $A$,
\begin{center}
  ${\vdash_{\HRPLA} A}$ \ iff\ \: ${\vdash_{\HLukA} A}$ \ iff\ \:
  ${\vdash_{\GthreeLukA+\textnormal{(\L cut)}} A}$ \ iff\\[\smallskipamount]
  ${\vdash_{\GzeroLukA+\textnormal{(\L cut)}} A}$ \ iff\ \:
  ${\vdash_{\GLukA+\textnormal{(\L cut)}} A}$.
\end{center}
\end{ltheorem}

In proving these theorems, we will use the cumulative cancellation rules
(cf. \cite[Section 4.1]{CiabattoniMetcalfe2003} and 
\cite[Section 4.3.5]{MOG2009})\footnote{
  The (noncumulative) cancellation rule was introduced 
  in \cite{CiabattoniMetcalfe2003}
  as a variant of the cut rule 
  to establish cut elimination 
  for the propositional fragment of the calculus $\GLukA$ 
  via elimination of the cancellation rule.
}
$$\dfrac{\mathcal{G}\,|\, \Gamma\Rightarrow\Delta\,|\, \Gamma, C\Rightarrow C,\Delta}
  {\mathcal{G}\,|\, \Gamma\Rightarrow\Delta}
   \text{\small \arraycolsep=0.5mm 
     $\left\{
       \begin{array}{ll}
         \text{(ccan)}, & \text{where $C$ is an $\RPLA$-formula,}\\
         \text{(\L ccan)}, & \text{where $C$ is an $\LukA$-formula;}
       \end{array} \right.$
   }
$$
and the calculi $\HcRPLA$ and $\HcLukA$ that are obtained from $\HRPLA$ and
$\HLukA$, respectively, thus:
$t$ in the axiom schemes ($\forall$1) and ($\exists$1) is taken to be 
a closed term, and the inference rule (gen) is replaced by the rule
$$\dfrac{\RePl{A}{x}{a}}{\forall x A}~(\widehat{\text{gen}}),$$
where $a$ is a parameter not occurring in~$A$.

\begin{llemma} \label{Lem:HilbCalcEquForSent}
For any $\RPLA$-sentence $A$,
\,${\vdash_{\HRPLA} A}$  iff  \:${\vdash_{\HcRPLA} A}$.
\ For any $\LukA$-sentence $B$,
\,${\vdash_{\HLukA} B}$  iff  \:${\vdash_{\HcLukA} B}$.
\end{llemma}

We omit the proof of Lemma \ref{Lem:HilbCalcEquForSent}, because 
the proof is not complicated and does not differ from that of 
a similar assertion for appropriate variants of 
classical first-order Hilbert-type calculi.

We will also use the following translations and 
Lemma~\ref{Lem:GLukAPropVarsInsteadTruthConstImpGi}.

For an $\RPLA$-formula $A$ and a hypersequent $\mathcal{H}$, let
the translations ${A \mapsto \widetilde{A}}$ and
${\mathcal{H} \mapsto \widetilde{\mathcal{H}}}$ replace each truth constant
$\bar{r} \neq \bar{0}$ by a unique propositional variable~$p_{\bar{r}}$.

\begin{llemma}  \label{Lem:GLukAPropVarsInsteadTruthConstImpGi}
Let $\mathcal{H}$ be a hypersequent not containing semipropositional 
variables, and $\widetilde{\mathcal{H}}$ be the $\LukA$-hypersequent that
results from $\mathcal{H}$ by applying the above translation.
If \:${\vdash_{\GLukA} \widetilde{\mathcal{H}}}$, then
\,${\vdash_{\GiLukA} \mathcal{H}}$ for all ${i=0,1,\hat{1},3}$.
\end{llemma}

\begin{proof}[\textsc{Proof}]
Fix any ${i \in \{0,1,\hat{1},3\}}$.
By Theorem \ref{Th:GzeroGoneGonezGthreeConsExtGLukA},
from \,${\vdash_{\GLukA} \widetilde{\mathcal{H}}}$ it follows that
there exists a $\GiLukA$-proof $D$ of $\widetilde{\mathcal{H}}$.
For each truth constant $\bar{r} \neq \bar{0}$ occurring in $\mathcal{H}$, 
we replace the propositional variable $p_{\bar{r}}$ by $\bar{r}$ in $D$; 
thus we get a $\GiLukA$-proof of $\mathcal{H}$.
\end{proof}

\begin{proof}[\textsc{Proof of Theorem~\ref{Th:GzeroGthreeHRPLA}}]
For any $\RPLA$-sentence $A$, the following implications hold:
\begin{center}
$\vdash_{\HRPLA} A$
$\stackrel{\ref{Lem:HilbCalcEquForSent}}{\implies}$
$\vdash_{\HcRPLA} A$                            
$\stackrel{\ref{Lem:HcRPLAImpGthreeCut}}{\implies}$
$\vdash_\textnormal{$\GthreeLukA$+(cut)} A$     
$\stackrel{\ref{Lem:GthreeCutImpGzeroCcan}}{\implies}$\\ 
$\vdash_\textnormal{$\GzeroLukA$+(ccan)} A$     
$\stackrel{\ref{Lem:GzeroCcanImpGzeroCut}}{\implies}$
$\vdash_\textnormal{$\GzeroLukA$+(cut)} A$.
\end{center}
Over each of these implications, there is a number of the lemma that
verifies it. 
\end{proof}

\begin{llemma}  \label{Lem:HcRPLAImpGthreeCut}
For any $\RPLA$-formula $A$,
if \:${\vdash_{\HcRPLA} A}$, then \,${\vdash_\textnormal{$\GthreeLukA$+(cut)} A}$.
\end{llemma}

\begin{proof}[\textsc{Proof}]
The rule $(\widehat{\text{gen}})$ of $\HcRPLA$ is derivable in $\GthreeLukA$,
because the latter calculus contains the rule ${(\Rightarrow\forall)^3}$.

On the left in Figure \ref{fig:MpAdmGthreeCut}, we show how to get
the conclusion of the rule (mp) from its premises and 
the hypersequent ${\mathcal{H} = (A, A \to B \Rightarrow B)}$
using the rule (cut); 
and on the right, we show a $\GLukA$-proof of $\widetilde{\mathcal{H}}$.
\begin{figure}[!tb]
\begin{center} \small
\def\ScoreOverhang{0pt}
\AxiomC{$\Rightarrow A$;}
  \AxiomC{$\Rightarrow A \to B$;}
  \AxiomC{$A, A \to B \Rightarrow B$}
\RightLabel{(cut)}
\BinaryInfC{$A \Rightarrow B$}
\RightLabel{(cut)}
\BinaryInfC{$\Rightarrow B$}
\DisplayProof
\qquad
\AxiomC{$\widetilde{A} \Rightarrow \widetilde{A}$;}
\AxiomC{$\widetilde{B} \Rightarrow \widetilde{B}$}
\RightLabel{(mix)}
\BinaryInfC{$\widetilde{A},\widetilde{B} \Rightarrow \widetilde{A},\widetilde{B}$}
\RightLabel{$(\to\:\Rightarrow)$}
\UnaryInfC{$\widetilde{A}, \widetilde{A} \to \widetilde{B} \Rightarrow \widetilde{B}$}
\DisplayProof
\caption{Proofs for showing the derivability of (mp) in $\GthreeLukA$+(cut).}
\label{fig:MpAdmGthreeCut}
\end{center}
\end{figure}
By Lemma \ref{Lem:GLukAPropVarsInsteadTruthConstImpGi}, we have
\,${\vdash_{\GthreeLukA} \mathcal{H}}$.
So (mp) is derivable in $\GthreeLukA$+(cut).

To conclude the proof, it is sufficient to establish that 
all axioms of $\HcRPLA$ are $\GthreeLukA$-provable.

Let $L$ be an instance of one of the axiom schemes (\L1)--(\L4), say,
$$L = (A \to B) \to ((B \to C) \to (A \to C))$$ 
for some $\RPLA$-formulas $A,B,C$.
Take the propositional $\LukA$-formula 
$$L' = (p_1 \to p_2) \to ((p_2 \to p_3) \to (p_1 \to p_3)),$$
where $p_1,p_2,p_3$ are distinct propositional variables.
Since \,${\vDash L'}$ and
$\GLukA$ is complete for quantifier-free $\LukA$-hypersequents
\cite[Theorem 6.24]{MOG2009}, 
we have a $\GLukA$-proof $D'$ of $L'$.
In $D'$ we replace propositional variables $p_1,p_2,p_3$ with 
$\widetilde{A},\widetilde{B},\widetilde{C}$, 
respectively, producing a $\GLukA$-proof of $\widetilde{L}$.
Then by Lemma \ref{Lem:GLukAPropVarsInsteadTruthConstImpGi}, we get
\,${\vdash_{\GthreeLukA} L}$.

Any instance of any of the axiom schemes (tc1) and (tc2) 
is $\GthreeLukA$-provable, 
because it is valid and $\GthreeLukA$ is complete for quantifier-free
hypersequents \cite[Proposition 1]{Ger201t}. 

Finally, let $Q$ be a quantifier axiom of $\HcRPLA$.
Then we can construct a $\GLukA$-proof of $\widetilde{Q}$.
Indeed, in the cases of ($\forall$1) and ($\exists$1), this is trivial;
in the case of ($\exists$2), such a $\GLukA$-proof is given 
in Figure \ref{fig:ExistsTwoProvGLukA} 
(where $a$ does not occur in $\widetilde{Q}$);
\begin{figure}[!tb]
\begin{center} \small
\def\ScoreOverhang{0pt}
\AxiomC{$\Rightarrow$\,;}
  \AxiomC{$\Rightarrow$}
  \LeftLabel{(wl)}
  \UnaryInfC{$\forall x (\widetilde{A} \to \widetilde{B}) \Rightarrow$\,;}

  \AxiomC{$\RePl{\widetilde{A}}{x}{a} \Rightarrow \RePl{\widetilde{A}}{x}{a}$;}
  \AxiomC{$\widetilde{B} \Rightarrow \widetilde{B}$}
  \RightLabel{(mix)}
  \BinaryInfC{$\RePl{\widetilde{A}}{x}{a}, \widetilde{B} \Rightarrow \RePl{\widetilde{A}}{x}{a}, \widetilde{B}$}
  \RightLabel{$(\to\:\Rightarrow)$}
  \UnaryInfC{$\RePl{\widetilde{A}}{x}{a} \to \widetilde{B}, \RePl{\widetilde{A}}{x}{a} \Rightarrow \widetilde{B}$}
  \RightLabel{$(\forall\Rightarrow)$}  
  \UnaryInfC{$\forall x (\widetilde{A} \to \widetilde{B}), \RePl{\widetilde{A}}{x}{a} \Rightarrow \widetilde{B}$}
  \RightLabel{$(\exists\Rightarrow)$}  
  \UnaryInfC{$\forall x (\widetilde{A} \to \widetilde{B}), \exists x \widetilde{A} \Rightarrow \widetilde{B}$}
\RightLabel{$(\Rightarrow\:\to)$}
\BinaryInfC{$\forall x (\widetilde{A} \to \widetilde{B}) \Rightarrow (\exists x \widetilde{A} \to \widetilde{B})$}
\RightLabel{$(\Rightarrow\:\to)$}
\BinaryInfC{$\Rightarrow \forall x (\widetilde{A} \to \widetilde{B}) \to (\exists x \widetilde{A} \to \widetilde{B})$}
\DisplayProof
\caption{A $\GLukA$-proof of $\widetilde{(\exists2)}$.}
\label{fig:ExistsTwoProvGLukA}
\end{center}
\end{figure}
and in the case of ($\forall$2), a $\GLukA$-proof of $\widetilde{Q}$ 
is constructed similarly.
Hence \,${\vdash_{\GthreeLukA} Q}$ 
by Lemma \ref{Lem:GLukAPropVarsInsteadTruthConstImpGi}.
\end{proof}

\begin{llemma}  \label{Lem:GthreeCutImpGzeroCcan}
If \:${\vdash_\textnormal{$\GthreeLukA$+(cut)} \mathcal{H}_0}$,
then \,${\vdash_\textnormal{$\GzeroLukA$+(ccan)} \mathcal{H}_0}$.
\end{llemma}

\begin{proof}[\textsc{Proof}]
It suffices to show that all the rules of $\GthreeLukA$+(cut) are admissible for
$\GzeroLukA$+(ccan).

1. Let us demonstrate that the rules $\textnormal{(ew)}^0$, $\textnormal{(ec)}^0$,
and $\textnormal{(split)}^0$ are hp-admissible for $\GzeroLukA$+(ccan), and
the rule $\textnormal{(mix)}^0$ is admissible for $\GzeroLukA$+(ccan)
(these rules are formulated in Lemma~\ref{Lem:StructRulesAdmGzero}).

For $\textnormal{(ew)}^0$, $\textnormal{(ec)}^0$, and $\textnormal{(mix)}^0$,
these assertions are established just as in items 1, 2, and 5 of the proof of 
Lemma~\ref{Lem:StructRulesAdmGzero}.

The proof of the hp-admissibility of $\textnormal{(split)}^0$ 
for $\GzeroLukA$+(ccan) is similar to the proof of Lemma 7 in \cite{Ger2017},\footnote{ 
  See also Section \ref{secApp:SplitMixAdmGzero} 
  (of the appendix) on p.~\pageref{secApp:SplitMixAdmGzero}.
}
we only need to consider one more case.
As in the proof of Lemma 7 in \cite{Ger2017}, by induction on the height of
a proof $D_1$ of
\,${\mathcal{G}\,|\,\Gamma_1,\Gamma_2\Rightarrow\Delta_1,\Delta_2}$\,
(in $\GzeroLukA$+(ccan) now),
we show that $D_1$ can be transformed into a proof of 
\,${\mathcal{G}\,|\,\Gamma_1\Rightarrow\Delta_1\,|\,\Gamma_2\Rightarrow\Delta_2}$\,
whose height is not greater than the height of~$D_1$.
We add the case where the proof $D_1$ has the form:
\begin{center}
\def\ScoreOverhang{0pt}
\AxiomC{$D_0$} \noLine 
\UnaryInfC{$\mathcal{G} \,|\, \Gamma_1,\Gamma_2 \Rightarrow \Delta_1,\Delta_2
  \,|\, \Gamma_1,\Gamma_2,A \Rightarrow A,\Delta_1,\Delta_2$}
\RightLabel{(ccan).}
\UnaryInfC{$\mathcal{G} \,|\, \Gamma_1,\Gamma_2 \Rightarrow \Delta_1,\Delta_2$}
\DisplayProof
\end{center}
In this case, using the induction hypothesis twice, 
we split the two sequent occurrences 
distinguished in the lowest hypersequent in the proof $D_0$ 
to obtain a proof of
$$\mathcal{G} \,|\, \Gamma_1 \Rightarrow \Delta_1 \,|\, \Gamma_2 \Rightarrow \Delta_2
  \,|\, \Gamma_1 \Rightarrow \Delta_1
  \,|\, \Gamma_2,A \Rightarrow A,\Delta_2;$$
whence by the hp-admissible rule $\textnormal{(ec)}^0$,
we construct a proof of
$$\mathcal{G} \,|\, \Gamma_1 \Rightarrow \Delta_1 \,|\, \Gamma_2 \Rightarrow \Delta_2
  \,|\, \Gamma_2,A \Rightarrow A,\Delta_2;$$
and by (ccan), we get the desired proof of
\,${\mathcal{G}\,|\,\Gamma_1\Rightarrow\Delta_1\,|\,\Gamma_2\Rightarrow\Delta_2}$.

2. Let us establish the admissibility of the rules $(\textnormal{den}_1)$ and 
$(\textnormal{den}_0)$ for $\GzeroLukA$+(ccan) 
(these rules are formulated at the beginning of Section~\ref{sec:FurtherRelat}).
With the results of the preceding item, we do this
as in Lemmas \ref{Lem:DenOneAdmGzero} and \ref{Lem:DenZeroAdmGzero}, 
adding to item 2.2 of the proof of Lemma \ref{Lem:DenOneAdmGzero} 
one more case 2.2.6 where $\mathcal{R}$ is (ccan) and 
the proof $D$ (in $\GzeroLukA$+(ccan) now) looks like:
\begin{center}
\def\ScoreOverhang{0pt}
\AxiomC{$D_1$} \noLine 
\UnaryInfC{$\mathcal{H} \,|\, \Gamma_1, A, \SpV{p}_1 \Rightarrow A, \Delta_1$}
\RightLabel{(ccan).} 
\UnaryInfC{$\mathcal{H}$}
\DisplayProof
\end{center}
In this case, using the induction hypothesis, we transform $D_1$ into a proof of
$$\mathcal{H}' \,\big|\, \big[ \Gamma_1, A, \Pi_j \Rightarrow A, \Delta_1, \Sigma_j \big]_{j\in 1..n},$$
whence we get the desired proof of $\mathcal{H}'$ by $n$ applications of (ccan).

3. Now the admissibility for $\GzeroLukA$+(ccan) of each rule of $\GthreeLukA$
can be shown just as in the proof of Lemma \ref{Lem:IfDeniAdmGzero}.
Finally, (cut) is admissible for $\GzeroLukA$+(ccan).
Indeed, the conclusion of (cut) is obtained from its premises thus:
\begin{center}
\def\ScoreOverhang{0pt}
\AxiomC{$\mathcal{G}\,|\, \Gamma_1 \Rightarrow C,\Delta_1$;}
\AxiomC{$\mathcal{G}\,|\, \Gamma_2,C \Rightarrow \Delta_2$}
\RightLabel{$\textnormal{(mix)}^0$}
\BinaryInfC{$\mathcal{G}\,|\, \Gamma_1,\Gamma_2,C \Rightarrow C,\Delta_1,\Delta_2$}
\RightLabel{$\textnormal{(ew)}^0$}
\UnaryInfC{$\mathcal{G}\,|\, \Gamma_1,\Gamma_2 \Rightarrow \Delta_1,\Delta_2 \,|\, \Gamma_1,\Gamma_2,C \Rightarrow C,\Delta_1,\Delta_2$}
\RightLabel{(ccan),}
\UnaryInfC{$\mathcal{G}\,|\, \Gamma_1,\Gamma_2 \Rightarrow \Delta_1,\Delta_2$}
\DisplayProof
\end{center}
$\textnormal{(mix)}^0$ and $\textnormal{(ew)}^0$ being
admissible for $\GzeroLukA$+(ccan) by item 1 of the present proof.
\end{proof}

\begin{llemma}  \label{Lem:GzeroCcanImpGzeroCut}
\ ${\vdash_\textnormal{$\GzeroLukA$+(ccan)} \mathcal{H}}$  iff
\:${\vdash_\textnormal{$\GzeroLukA$+(cut)} \mathcal{H}}$.
\end{llemma}

\begin{proof}[\textsc{Proof}]
For the left-to-right direction, it is enough to establish that 
(ccan) is admissible for $\GzeroLukA$+(cut).
The conclusion of (ccan) is obtained from its premise and
the hypersequents $\Rightarrow$ and \,${\mathcal{H} = (C \Rightarrow C)}$\,
by rules, which are admissible for $\GzeroLukA$+(cut), as follows
(cf. \cite[Section 4.1]{CiabattoniMetcalfe2003}):
\begin{center}
\def\ScoreOverhang{0pt}
  \AxiomC{$\Rightarrow$\,;}
  \AxiomC{$C \Rightarrow C$}
  \LeftLabel{$\textnormal{(ew)}^0$,$(\Rightarrow\:\to)^0$}
  \BinaryInfC{$\Rightarrow C\to C$;}
  
  \AxiomC{$\mathcal{G}\,|\, \Gamma \Rightarrow \Delta \,|\, \Gamma,C \Rightarrow C,\Delta$}
  \RightLabel{$\textnormal{(ew)}^0$,$(\to\:\Rightarrow)^0$}
  \UnaryInfC{$\mathcal{G}\,|\, \Gamma, C\to C \Rightarrow \Delta$}

\RightLabel{$\textnormal{(ew)}^0$,(cut).}
\BinaryInfC{$\mathcal{G}\,|\, \Gamma \Rightarrow \Delta$}
\DisplayProof
\end{center}
The hypersequent $\Rightarrow$ is an axiom of $\GzeroLukA$.
The hypersequent $\widetilde{\mathcal{H}}$ is an axiom of $\GLukA$,
and hence by Lemma \ref{Lem:GLukAPropVarsInsteadTruthConstImpGi}, we get
\,${\vdash_{\GzeroLukA} \mathcal{H}}$.
Thus (ccan) is admissible for $\GzeroLukA$+(cut).

For the right-to-left direction,
it suffices to show that (cut) is admissible for $\GzeroLukA$+(ccan).
This is done in item 3 of the proof of Lemma \ref{Lem:GthreeCutImpGzeroCcan}.
\end{proof}

\begin{proof}[\textsc{Proof of Theorem~\ref{Th:GzeroGthreeHLukA}}]
Let $A$ be an $\LukA$-sentence.
By \cite[Theorem 2.4]{HajekParisShepherdson2000}, we have:
\,${\vdash_{\HRPLA} A}$ iff \,${\vdash_{\HLukA} A}$.
Also the following implications hold:
\begin{center}
$\vdash_{\HLukA} A$  
$\stackrel{\ref{Lem:HilbCalcEquForSent}\phantom{1}}{\implies}$
$\vdash_{\HcLukA} A$  
$\stackrel{\ref{Lem:HcRPLAImpGthreeCut}^{\approx}}{\implies}$
$\vdash_\textnormal{$\GthreeLukA$+(\L cut)} A$    
$\stackrel{\ref{Lem:GthreeCutImpGzeroCcan}^{\approx}}{\implies}$
$\vdash_\textnormal{$\GzeroLukA$+(\L ccan)} A$     
$\stackrel{\ref{Lem:GzeroCcanImpGzeroCut}^{\approx}}{\implies}$\\
$\vdash_\textnormal{$\GzeroLukA$+(\L cut)} A$
$\stackrel{\ref{Lem:GzeroImpGLukA}^{\approx}}{\implies}$
$\vdash_\textnormal{$\GLukA$+(\L cut)} A$
$\implies$  
$\vdash_{\HLukA} A$.
\end{center}
The first implication holds by Lemma \ref{Lem:HilbCalcEquForSent};
the last one, by \cite[Theorem 9]{BaazMetcalfe2010}.
Over each of the other implications, a number (with the symbol~$^{\approx}$) 
is given indicating that the implication is proved by analogy with
the lemma designated by the number.   
\end{proof}

\section{Conclusion}
\label{sec:Concl}

In the present article, we have established that the calculi $\GzeroLukA$ 
and $\GthreeLukA$ for the logic $\RPLA$ are equivalent and
are conservative extensions of the calculus $\GLukA$ for the logic $\LukA$
(see Theorems \ref{Th:GzeroGonezGthreeEquiv} and 
\ref{Th:GzeroGoneGonezGthreeConsExtGLukA}).

The crucial part of our argument is the syntactic proofs of the admissibility
for $\GzeroLukA$ of the nonstandard variants $(\textnormal{den}_1)$ and 
$(\textnormal{den}_0)$ of the density rule
(see Lemmas \ref{Lem:DenOneAdmGzero} and~\ref{Lem:DenZeroAdmGzero}).
These proofs can be easily adapted to show the admissibility 
of the nonstandard density rule for $\GzeroLukA$
(see Remark \ref{rem:DensityAdmGzero} on p.~\pageref{rem:DensityAdmGzero}).
The given proof of the admissibility of $(\textnormal{den}_1)$ for $\GzeroLukA$
provides an algorithm for transforming a proof of a hypersequent 
in $\GzeroLukA$+$(\textnormal{den}_1)$ into a proof of the same hypersequent 
in $\GzeroLukA$,
in other words, establishes elimination of $(\textnormal{den}_1)$ for
$\GzeroLukA$+$(\textnormal{den}_1)$;\footnote{
  Suppose that a rule $\mathcal{R}$ is not an inference rule of 
  a calculus $\mathfrak{C}$.
  It is said that \emph{elimination of $\mathcal{R}$ holds for} 
  $\mathfrak{C}$+$\mathcal{R}$ (as well as \emph{for} $\mathfrak{C}$)
  if a proof of an object in $\mathfrak{C}$+$\mathcal{R}$ can be 
  algorithmically transformed into a proof of the same object in $\mathfrak{C}$
  (cf., e.g., \cite{CiabattoniMetcalfe2003, MetcalfeMontagna2007, MOG2009, 
  MetcalfeTsinakis2017}). 
}
similarly with $(\textnormal{den}_0)$ and the nonstandard density rule.

Density elimination proofs are known for some calculi
(and for some classes of calculi), though for logics 
different from $\LukA$ and $\RPLA$;
see \cite{BaazZach2000, BaazCiabattoniFermuller2003, MetcalfeMontagna2007,
CiabattoniMetcalfe2008, MOG2009, 
Baldi2012, Baldi2014, Baldi2015a, Baldi2015b, MetcalfeTsinakis2017, Baldi2017}.
In all these works except \cite{BaazCiabattoniFermuller2003},
such proofs use the cut rule even if no application of it is in 
an initial formal proof. 
In \cite{BaazCiabattoniFermuller2003} the density elimination proof 
for a single-conclusion hypersequent calculus 
for first-order G{\"o}del logic does not introduce cuts
if no cuts are in an initial formal proof.
Recall that the cut rule is not admissible for $\GzeroLukA$
(see Proposition \ref{Pr:CutNotAdm}).
Our technique for proving the admissibility of $(\textnormal{den}_1)$ 
for $\GzeroLukA$ resembles 
the technique of \cite{BaazCiabattoniFermuller2003}
for proving density elimination,
but was rediscovered 
and elaborated for the multiple-conclusion calculus $\GzeroLukA$ 
for the logic $\RPLA$.

Further, the book \cite{MOG2009} on p.~134 says that
it is unclear whether density elimination can be obtained 
for the propositional fragment of the calculus $\GLukA$.
We have given such a density elimination proof 
for the calculus $\GzeroLukA$, which is a conservative extension of $\GLukA$;
and let us note that a $\GLukA$-proof can be algorithmically transformed into
a $\GzeroLukA$-proof of the same $\LukA$-hypersequent, and conversely
(see Theorem \ref{Th:GzeroConsExtGLukA} and its proof).
Moreover, to the best of our knowledge, 
the given proof is the first syntactic proof of 
density admissibility 
for a multiple-con\-clu\-sion hypersequent calculus in which 
neither the weakening rule nor the contraction rule is admissible.\footnote{
  The weakening and contraction rules are, respectively:
  \begin{center}
  $\dfrac{\mathcal{G}\,|\,\Gamma\Rightarrow\Delta}
         {\mathcal{G}\,|\,\Gamma,\Pi\Rightarrow\Sigma,\Delta}$ 
  \quad and \quad
  $\dfrac{\mathcal{G}\,|\,\Gamma,\Pi,\Pi\Rightarrow\Sigma,\Sigma,\Delta}
         {\mathcal{G}\,|\,\Gamma,\Pi\Rightarrow\Sigma,\Delta}$,
  \end{center}
  where $\mathcal{G}$ is any hypersequent, and
  $\Gamma$, $\Delta$, $\Pi$, and $\Sigma$ are any finite multisets of formulas
  of a language under consideration (see \cite[Section 4.3]{MOG2009}).
}

It would be nice to generalize
the density elimination technique used in \cite{BaazCiabattoniFermuller2003} 
and in the present article to as wide a syntactic class of hypersequent calculi
as possible.
Besides, notice that how complexity of formal proofs varies
has not been investigated for any density elimination proof.

Also in the given article, we have established that the calculi $\HLukA$,
$\GzeroLukA$+(\L cut), and $\GthreeLukA$+(\L cut) prove the same 
$\LukA$-sentences (see Theorem \ref{Th:GzeroGthreeHLukA}).
And we have shown that the provability of an $\RPLA$-sentence in $\HRPLA$ 
implies its provability in $\GthreeLukA$+(cut), which in turn implies 
its provability in $\GzeroLukA$+(cut) (see Theorem \ref{Th:GzeroGthreeHRPLA}).

Thus natural open questions are whether 
any $\RPLA$-sentence provable in $\GthreeLukA$+(cut) is provable in $\HRPLA$,
and a similar question for $\GzeroLukA$+(cut).
To answer both questions affirmatively, 
it suffices to establish that 
any $\RPLA$-sentence provable in $\GzeroLukA$+(cut) is provable in $\HRPLA$,
assuming that hypersequents do not contain semipropositional variables, i.e.,
are built up only from $\RPLA$-formulas.
With this preparation, it is worth trying to extend to $\RPLA$ 
the algebraic technique used in \cite{BaazMetcalfe2010}
for showing that any $\LukA$-sentence provable in $\GLukA$+(\L cut) 
is provable in $\HLukA$.

\begin{small}

\end{small}

\newpage
\appendix
\begin{center} \bf
APPENDIX
\end{center}

\section{The admissibility of the rules $\textnormal{(split)}^0$ and
         $\textnormal{(mix)}^0$ for $\GzeroLukA$}
\label{secApp:SplitMixAdmGzero}

Item 4 of the proof of Lemma \ref{Lem:StructRulesAdmGzero} says that
the proof of the hp-ad\-mis\-si\-bi\-lity of $\textnormal{(split)}^0$ 
for $\GzeroLukA$ is very similar to the proof of Lemma 7 in \cite{Ger2017}.
Besides, in item 1 of the proof of Lemma \ref{Lem:GthreeCutImpGzeroCcan},
we extend the proof of the hp-admissibility of $\textnormal{(split)}^0$ 
for $\GzeroLukA$ with a new case and obtain the proof of
the hp-admissibility of $\textnormal{(split)}^0$ for $\GzeroLukA$+(ccan).

Next, item 5 of the proof of Lemma \ref{Lem:StructRulesAdmGzero} says that
the proof of the admissibility of $\textnormal{(mix)}^0$ for $\GzeroLukA$ 
can be easily obtained from the proof of Lemma 8 in \cite{Ger2017} 
by identifying the notion of a completable ancestor of a sequent occurrence
with the notion of an ancestor of a sequent occurrence
(the former notion being used in \cite{Ger2017}).

Below we give the proof of the hp-admissibility of $\textnormal{(split)}^0$ 
for $\GzeroLukA$ and the proof of the admissibility of $\textnormal{(mix)}^0$ 
for $\GzeroLukA$, adapting the mentioned proofs in \cite{Ger2017}
(and correcting some inaccuracies introduced in \cite{Ger2017} by
a translator of the original Russian article).

\begin{llemma} 
The following rule is hp-admissible for $\GzeroLukA$:
$$\dfrac{\mathcal{G}\,|\,\Gamma_1,\Gamma_2\Rightarrow\Delta_1,\Delta_2}
        {\mathcal{G}\,|\,\Gamma_1\Rightarrow\Delta_1\,|\,\Gamma_2\Rightarrow\Delta_2}~\textnormal{(split)}^0.$$
\end{llemma}

\begin{proof}[\textsc{Proof}]
Let 
$$\mathcal{H}_1 = (\mathcal{G}\,|\,\Gamma_1,\Gamma_2\Rightarrow\Delta_1,\Delta_2),
  \quad 
  \mathcal{H}_2 = (\mathcal{G}\,|\,\Gamma_1\Rightarrow\Delta_1\,|\,\Gamma_2\Rightarrow\Delta_2).$$

Using induction on the height of a ($\GzeroLukA$-)proof $D_1$ of $\mathcal{H}_1$,
we show that $D_1$ can be transformed into a proof of $\mathcal{H}_2$
whose height is not greater than the height of~$D_1$.

1. If $\mathcal{H}_1$ is an axiom, then it is easy to see that 
$\mathcal{H}_2$ is an axiom too.

2. Let the lowest hypersequent $\mathcal{H}_1$ in $D_1$ be the conclusion of
an application $R$ of a rule~$\mathcal{R}$.
We consider the case where $\mathcal{R}$ is ${(\to\:\Rightarrow)^0}$;
the remaining cases are similar.

2.1. Suppose that the principal sequent occurrence in the application $R$ is in
the distinguished occurrence of $\mathcal{G}$ in $\mathcal{H}_1$.
Then the premise $\mathcal{H}_0$ of the application $R$ has the form
\,${\mathcal{G}_0\,|\,\Gamma_1,\Gamma_2\Rightarrow\Delta_1,\Delta_2}$.
By the induction hypothesis for the proof of $\mathcal{H}_0$ 
(which is a subtree of the proof tree $D_1$), we can construct a proof of  
\,${\mathcal{G}_0\,|\,\Gamma_1\Rightarrow\Delta_1\,|\,\Gamma_2\Rightarrow\Delta_2}$.
Applying the rule $\mathcal{R}$, we obtain the required proof of~$\mathcal{H}_2$.

2.2. Suppose that the principal sequent occurrence in the application $R$ is
the distinguished occurrence of
\,${\Gamma_1,\Gamma_2\Rightarrow\Delta_1,\Delta_2}$\, in $\mathcal{H}_1$.
For definiteness we assume that the principal occurrence of a formula
${A_1 \to B_1}$ in the application $R$ is in $\Gamma_1$.
Then ${\Gamma_1 = (\Gamma_1', A_1 \to B_1)}$ for some $\Gamma_1'$.
The proof $D_1$ has the form
\begin{center}
\def\ScoreOverhang{0pt}
\AxiomC{$D_0$} \noLine 
\UnaryInfC{$\genfrac{}{}{0pt}{}{\displaystyle 
    \mathcal{G} \,|\, \Gamma_1', A_1 \to B_1, \Gamma_2 \Rightarrow \Delta_1, \Delta_2
    \,|\, \Gamma_1', \Gamma_2 \Rightarrow \Delta_1, \Delta_2 
  }{\displaystyle 
    |\, \Gamma_1', B_1, \Gamma_2 \Rightarrow A_1, \Delta_1, \Delta_2
  }$}
\RightLabel{$(\to\:\Rightarrow)^0$.}
\UnaryInfC{$\mathcal{G}\,|\, \Gamma_1', A_1 \to B_1, \Gamma_2 \Rightarrow \Delta_1, \Delta_2$}
\DisplayProof
\end{center}

We use the induction hypothesis twice and split all the three sequent occurrences
that are distinguished in the lowest hypersequent in the proof $D_0$.
We obtain a proof of the hypersequent 
\begin{gather*}
 \mathcal{G} \,|\, \Gamma_1', A_1 \to B_1 \Rightarrow \Delta_1 
   \,|\, \Gamma_2 \Rightarrow \Delta_2 
   \,|\, \Gamma_1' \Rightarrow \Delta_1
   \,|\, \Gamma_2 \Rightarrow \Delta_2 \\ 
 |\, \Gamma_1', B_1 \Rightarrow A_1, \Delta_1
   \,|\, \Gamma_2 \Rightarrow \Delta_2.
\end{gather*}

From this hypersequent, we eliminate two occurrences of 
${\Gamma_2 \Rightarrow \Delta_2}$ with the help of 
the hp-admissible rule $\textnormal{(ec)}^0$. 
We get a proof $D_0'$ of the hypersequent 
$$\mathcal{G} \,|\, \Gamma_1', A_1 \to B_1 \Rightarrow \Delta_1 
  \,|\, \Gamma_1' \Rightarrow \Delta_1
  \,|\, \Gamma_1', B_1 \Rightarrow A_1, \Delta_1
  \,|\, \Gamma_2 \Rightarrow \Delta_2.$$

Finally, we apply the rule ${(\to\:\Rightarrow)^0}$ to the lowest hypersequent
in $D_0'$ and obtain the required proof of 
\,${\mathcal{G} \,|\, \Gamma_1', A_1 \to B_1 \Rightarrow \Delta_1 
   \,|\, \Gamma_2 \Rightarrow \Delta_2}$.
\end{proof}

\begin{llemma} 
The following rule is admissible for $\GzeroLukA$:
$$\dfrac{\mathcal{G}\,|\,\Gamma_1\Rightarrow\Delta_1;  \quad  \mathcal{G}\,|\,\Gamma_2\Rightarrow\Delta_2}
       {\mathcal{G}\,|\,\Gamma_1,\Gamma_2\Rightarrow\Delta_1,\Delta_2}~\textnormal{(mix)}^0.$$
\end{llemma}

\begin{proof}[\textsc{Proof}]
Let 
\begin{gather*}
 \mathcal{H}_1 = (\mathcal{G}\,|\,\Gamma_1\Rightarrow\Delta_1), \quad
 \mathcal{H}_2 = (\mathcal{G}\,|\,\Gamma_2\Rightarrow\Delta_2), \\
 \mathcal{H}_3 = (\mathcal{G}\,|\,\Gamma_1,\Gamma_2\Rightarrow\Delta_1,\Delta_2).
\end{gather*}
We suppose that \,${\vdash_{\GzeroLukA} \mathcal{H}_1}$ and
\,${\vdash_{\GzeroLukA} \mathcal{H}_2}$, and
show that \,${\vdash_{\GzeroLukA} \mathcal{H}_3}$.
Let $D_1$ be a ($\GzeroLukA$-)proof of $\mathcal{H}_1$ such that
no proper parameter 
from $D_1$ occurs in ${\Gamma_2\Rightarrow\Delta_2}$.

We obtain a proof search tree $D_3^0$ for $\mathcal{H}_3$ as follows.
In $D_1$, for each occurrence $\mathcal{S}$ of a sequent of the form
${\Pi_1\Rightarrow\Sigma_1}$,
if $\mathcal{S}$ is an ancestor of the distinguished occurrence of
the sequent ${\Gamma_1\Rightarrow\Delta_1}$ in the root of $D_1$, 
then we replace $\mathcal{S}$ by an occurrence $\mathcal{S}'$ of the sequent
${\Pi_1,\Gamma_2 \Rightarrow \Sigma_1,\Delta_2}$.
We also mark $\mathcal{S}'$ if
$\mathcal{S}$ is an atomic sequent occurrence in a leaf of $D_1$.
Let $\mathcal{S}_i$, ${i=0,\ldots,l-1}$, be all distinct marked sequent
occurrences in~$D_3^0$.

We expand $D_3^0$, proceeding for each ${i=0,\ldots,l-1}$ as follows. 

(0) Let $\mathcal{S}_i$ be an occurrence of a sequent of the form
${\Pi_1,\Gamma_2 \Rightarrow\! \Sigma_1,\Delta_2}$.

(1) We construct a proof $D_2$ of $\mathcal{H}_2$ such that no proper parameter 
from $D_2$ occurs in~$D_3^i$.

(2) We obtain a proof search tree $\widehat{D}_2$ for 
${\mathcal{G} \,|\, \Pi_1,\Gamma_2 \Rightarrow \Sigma_1,\Delta_2}$ thus:
in $D_2$, for each occurrence of a sequent of the form
${\Pi_2\Rightarrow\Sigma_2}$, if this occurrence is 
an ancestor of the distinguished occurrence of the sequent
${\Gamma_2\Rightarrow\Delta_2}$ in the root of $D_2$, 
then we replace this occurrence by ${\Pi_1,\Pi_2 \Rightarrow \Sigma_1,\Sigma_2}$.

(3) We expand each branch of $D_3^i$ containing the occurrence $\mathcal{S}_i$
as follows: we identify the top node of this branch, which represents 
on occurrence of a hypersequent of the form
${\mathcal{G} \,|\, \Pi_1,\Gamma_2 \Rightarrow \Sigma_1,\Delta_2 
  \,|\, \mathcal{H}}$ for some $\mathcal{H}$, 
with the root of the tree 
obtained from $\widehat{D}_2$ by appending ``$|\,\mathcal{H}$''
to each node hypersequent. 
By $D_3^{i+1}$ we denote the tree resulting from this expansion of~$D_3^i$.

It is not difficult to see that the tree $D_3^l$ is a proof search tree for 
$\mathcal{H}_3$. 
It remains to show that $D_3^l$ is a proof.

We consider an arbitrary leaf $L_3$ of $D_3^l$ and show that $L_3$ is an axiom.
Given $L_3$, we find a unique leaf $L_1$ of $D_1$ that transforms 
into a leaf of $D_3^0$ that, in turn, transforms (in expanding $D_3^0$) into
a node of $D_3^l$ belonging to the same branch as~$L_3$.

Let ${\Pi_{1,i}\Rightarrow\Sigma_{1,i}}$, ${i \in I}$, be all atomic sequents
whose occurrences in $L_1$ are 
ancestors of the distinguished occurrence of 
${\Gamma_1\Rightarrow\Delta_1}$ in the root of $D_1$.
By the construction of $D_3^l$, for each ${i \in I}$,
there exist a proof $D_2^i$ of $\mathcal{H}_2$ and 
its leaf $L_2^i$ such that, for each ${j \in J_i}$,
an atomic sequent 
${\Pi_{1,i}, \Pi^i_{2,j} \Rightarrow \Sigma_{1,i}, \Sigma^i_{2,j}}$
occurs in $L_3$, where
${\Pi^i_{2,j}\Rightarrow\Sigma^i_{2,j}}$, ${j \in J_i}$, are all atomic sequents
whose occurrences in $L_2^i$ are 
ancestors of the distinguished occurrence of ${\Gamma_2\Rightarrow\Delta_2}$ 
in the root of~$D_2^i$.

In addition, $L_3$ contains all atomic sequents $S_{1,k}$, ${k \in K}$,
whose occurrences in $L_1$ are ancestors of sequent occurrences in
the distinguished occurrence of $\mathcal{G}$ in the root of~$D_1$.

Finally, for each ${i \in I}$, the leaf $L_3$ contains all atomic sequents 
$S^i_{2,m}$, ${m \in M_i}$, whose occurrences in $L_2^i$ are 
ancestors of sequent occurrences in the distinguished occurrence of 
$\mathcal{G}$ in the root of~$D_2^i$.

The leaf $L_1$ of the proof $D_1$ is an axiom and 
contains exactly the following atomic sequents: 
${\Pi_{1,i}\Rightarrow\Sigma_{1,i}}$ for each ${i \in I}$ \,and\,
$S_{1,k}$ for each ${k \in K}$.
For each ${i \in I}$, the leaf $L_2^i$ of the proof $D_2^i$ is an axiom and
contains exactly the following atomic sequents: 
${\Pi^i_{2,j}\Rightarrow\Sigma^i_{2,j}} $ for each ${j \in J_i}$ \,and\, 
$S^i_{2,m}$ for each ${m \in M_i}$.
Therefore, the leaf $L_3$ of $D_3^l$, 
which contains the above-mentioned atomic sequents, is an axiom too.
\end{proof}

\section{The soundness of the nonstandard density rule}
\label{secApp:DenSound}

Remark \ref{rem:DensitySound} on p.~\pageref{rem:DensitySound} says that 
the rule 
$$\dfrac{\mathcal{G} \,|\, \Gamma, \SpV{p} \Rightarrow \Delta \,|\, \Pi \Rightarrow \SpV{p}, \Sigma}
  {\mathcal{G} \,|\, \Gamma, \Pi \Rightarrow \Delta, \Sigma}~(\text{den}),$$
(1)~is unsound if $\SpV{p}$ is a propositional variable
not occurring in the conclusion,
but (2)~becomes sound if we expand the notion of a hypersequent by 
special variables interpreted by any real numbers, and 
require $\SpV{p}$ to be such a variable not occurring in the conclusion.

\bigskip

Let us prove (1). 
Recall that the propositional variable $\SpV{p}$ is interpreted by
any real number in $[0,1]$.
Consider the following application of (den):
$$\dfrac{\SpV{p} \Rightarrow \bar{0},\bar{0} \,|\, \bar{0} \Rightarrow \SpV{p}}
  {\bar{0} \Rightarrow \bar{0},\bar{0}}.$$
The premise of this application is valid
(because \,${\vDash (\bar{0} \Rightarrow \SpV{p})}$),
but its conclusion is not valid.
\hfill$\Box$

\bigskip

Now let us prove (2), i.e., that 
\,${\vDash (\mathcal{G} \,|\, \Gamma, \SpV{p} \Rightarrow \Delta \,|\, \Pi \Rightarrow \SpV{p}, \Sigma)}$
implies 
\,${\vDash (\mathcal{G} \,|\, \Gamma, \Pi \Rightarrow \Delta, \Sigma)}$
under the specified restriction on $\SpV{p}$.
To make this proof shorter, assume harmlessly that 
the hypersequent $\mathcal{G}$ is empty.

(a)~$\nvDash (\Gamma, \Pi \Rightarrow \Delta, \Sigma)$ 
  $\iff$ for some hs-interpretation $M$ and $M$-valuation~$\nu$,
  \ $\MsV{\Delta}_{M,\nu} - \MsV{\Gamma}_{M,\nu} < 
     \MsV{\Pi}_{M,\nu} - \MsV{\Sigma}_{M,\nu}$  $\iff$
  (by the density of the set of all real numbers)
  for some hs-interpretation $M$, $M$-valuation $\nu$, and real number~$\xi$,
  \ $\MsV{\Delta}_{M,\nu} - \MsV{\Gamma}_{M,\nu} < \xi-1 < 
     \MsV{\Pi}_{M,\nu} - \MsV{\Sigma}_{M,\nu}$.

(b)~$\nvDash (\Gamma, \SpV{p} \Rightarrow \Delta \,|\, \Pi \Rightarrow \SpV{p}, \Sigma)$ 
  $\iff$  for some hs-interpretation $M'$ and $M'$-valuation~$\nu'$,
  \ $\MsV{\Delta}_{M',\nu'} - \MsV{\Gamma}_{M',\nu'} < |\SpV{p}|_{M'} -1 < 
     \MsV{\Pi}_{M',\nu'} - \MsV{\Sigma}_{M',\nu'}$.

It is easy to see that (a) implies (b): take $\nu'=\nu$ 
and define $M'$ to be the same as $M$ but set ${|\SpV{p}|_{M'} = \xi}$.
\hfill$\Box$

\section{The admissibility of the nonstandard density rule for $\GzeroLukA$}
\label{secApp:DenAdmGzero}

Remark \ref{rem:DensityAdmGzero} on p.~\pageref{rem:DensityAdmGzero} says that
the proofs of Lemmas \ref{Lem:DenOneAdmGzero} and \ref{Lem:DenZeroAdmGzero}
can be easily combined to establish the admissibility for $\GzeroLukA$ 
of the rule
$$\dfrac{\mathcal{G} \,|\, \Gamma, \SpV{p} \Rightarrow \Delta \,|\, \Pi \Rightarrow \SpV{p}, \Sigma}
  {\mathcal{G} \,|\, \Gamma, \Pi \Rightarrow \Delta, \Sigma}~(\text{den}),$$
provided the notion of a hypersequent is expanded by 
special variables interpreted by any real numbers, and 
$\SpV{p}$ is such a variable not occurring in the conclusion.

Let us prove the next lemma on 
the admissibility of a generalization of (den) for $\GzeroLukA$, 
denoting by $\SpV{p}$ a special variable that can assume any real values
under hs-interpretations.

\begin{llemma}[admissibility of a generalization of (den) for $\GzeroLukA$]
Suppose that ${m\geqslant 1}$, ${n\geqslant 1}$, 
\begin{gather*}
{\mathcal{H} = \Big(\, \mathcal{G} \,\big|\, \big[ \Gamma_i, \SpV{p} \Rightarrow \Delta_i \big]_{i\in 1..m} \,\big|\, \big[ \Pi_j \Rightarrow \SpV{p}, \Sigma_j \big]_{j\in 1..n} \,\Big)},\\
{\mathcal{H}' = \Big(\, \mathcal{G} \,\big|\, \big[ \Gamma_i, \Pi_j \Rightarrow \Delta_i, \Sigma_j \big]^{i\in 1..m}_{j\in 1..n} \,\Big)},
\end{gather*}
$\SpV{p}$ does not occur in $\mathcal{H}'$, 
and \,${\vdash_{\GzeroLukA} \mathcal{H}}$.
Then \,${\vdash_{\GzeroLukA} \mathcal{H}'}$.
\end{llemma}

\begin{proof}[\textsc{Proof}]
Take a ($\GzeroLukA$-)proof $D$ of $\mathcal{H}$ and
proceed by induction on the height of~$D$.

1. Suppose that $\mathcal{H}$ is an axiom; i.e., ${\vDash \mathcal{H}_{at}}$.
Without loss of generality we assume that
$$\mathcal{H}_{at} = \Big(\, \mathcal{G}_{at} \,\big|\, 
  \big[ \Gamma_i, \SpV{p} \Rightarrow \Delta_i \big]_{i\in 1..k} \,\big|\, 
  \big[ \Pi_j \Rightarrow \SpV{p}, \Sigma_j \big]_{j\in 1..l} \,\Big),$$
where \,${0 \leqslant k \leqslant m}$ and \,${0 \leqslant l \leqslant n}$.
We put \,${\mathcal{H}'_{at} = (\mathcal{H}')_{at}}$.
Consider the following cases 1.1--1.4.

\textsl{Case 1.1: ${k \neq 0}$ and ${l \neq 0}$.}
We have
$${\mathcal{H}'_{at} = \Big(\, \mathcal{G}_{at} \,\big|\, \big[ \Gamma_i, \Pi_j \Rightarrow \Delta_i, \Sigma_j \big]^{i\in 1..k}_{j\in 1..l} \,\Big)}.$$

We want to show that ${\vDash \mathcal{H}'_{at}}$.
Suppose otherwise; i.e.,
for some hs-in\-ter\-pre\-ta\-tion $M$ and some $M$-valuation $\nu$, 
there is no true sequent in $\mathcal{G}_{at}$, and 
for all ${i\in 1..k}$ and ${j\in 1..l}$,
$$\MsV{\Delta_i}_{M,\nu} - \MsV{\Gamma_i}_{M,\nu} < \MsV{\Pi_j}_{M,\nu} - \MsV{\Sigma_j}_{M,\nu}.$$
By the density of the set $\mathbb{R}$ of all real numbers, 
there exists ${\xi \in \mathbb{R}}$ such that, 
for all ${i\in 1..k}$ and ${j\in 1..l}$,
$$\MsV{\Delta_i}_{M,\nu} - \MsV{\Gamma_i}_{M,\nu} < \xi-1 < \MsV{\Pi_j}_{M,\nu} - \MsV{\Sigma_j}_{M,\nu}.$$

Define an hs-interpretation $M_1$ to be like $M$, but set
${|\SpV{p}|_{M_1} = \xi}$.
Since $\SpV{p}$ does not occur in $\mathcal{G}_{at}$, 
$\Gamma_i$, $\Delta_i$ (${i\in 1..k}$), $\Pi_j$, $\Sigma_j$ (${j\in 1..l}$), 
we see that no sequent in $\mathcal{H}_{at}$ is true 
under the hs-interpretation $M_1$ and $M_1$-valuation $\nu$.
Hence ${\nvDash \mathcal{H}_{at}}$, a contradiction.

Therefore ${\vDash \mathcal{H}'_{at}}$, and so $\mathcal{H}'$ is an axiom.

\textsl{Case 1.2: ${k = 0}$ and ${l \neq 0}$.}
Then
$${\mathcal{H}_{at} = \Big(\, \mathcal{G}_{at} \,\big|\, 
    \big[ \Pi_j \Rightarrow \SpV{p}, \Sigma_j \big]_{j\in 1..l} \,\Big)}$$
and \,${\mathcal{H}'_{at} = \mathcal{G}_{at}}$.
Since $\SpV{p}$ does not occur in $\mathcal{G}_{at}$, $\Pi_j$, $\Sigma_j$
(${j\in 1..l}$), and 
hs-interpretations can take $\SpV{p}$ to negative real numbers whose 
absolute values are arbitrarily large, 
we conclude that 
${\vDash \mathcal{H}_{at}}$ implies ${\vDash \mathcal{G}_{at}}$.
Thus ${\vDash \mathcal{H}'_{at}}$ and $\mathcal{H}'$ is an axiom.

\textsl{Case 1.3: ${k \neq 0}$ and ${l = 0}$.}
Then
$${\mathcal{H}_{at} = \Big(\, \mathcal{G}_{at} \,\big|\, 
    \big[ \Gamma_i, \SpV{p} \Rightarrow \Delta_i \big]_{i\in 1..k} \,\Big)}$$
and \,${\mathcal{H}'_{at} = \mathcal{G}_{at}}$.
Since $\SpV{p}$ does not occur in  $\mathcal{G}_{at}$, $\Gamma_i$, $\Delta_i$ 
(${i\in 1..k}$), and $\SpV{p}$ can assume arbitrarily large values 
under hs-interpretations, we see that
${\vDash \mathcal{H}_{at}}$ implies ${\vDash \mathcal{G}_{at}}$.
So ${\vDash \mathcal{H}'_{at}}$ and $\mathcal{H}'$ is an axiom.

\textsl{Case 1.4: ${k = 0}$ and ${l = 0}$.}
Then \,${\mathcal{H}_{at} = \mathcal{G}_{at} = \mathcal{H}'_{at}}$.
Thus ${\vDash \mathcal{H}_{at}}$ means that
${\vDash \mathcal{H}'_{at}}$ and $\mathcal{H}'$ is an axiom.

2. It remains to consider the case where the root hypersequent $\mathcal{H}$ 
in $D$ is the conclusion of a rule application.
But the argument for this case can be obtained from
item 2 of the proof of Lemma \ref{Lem:DenOneAdmGzero} 
by replacing $\SpV{p}_1$ with~$\SpV{p}$.
\end{proof}

\end{document}